\documentclass[12pt]{article}
\usepackage{arxiv}
\usepackage{eurosym}
\usepackage{amsmath,amsthm,amsopn,amstext,amscd,amsfonts,amssymb}
\usepackage[english]{babel}
\usepackage{graphicx}
\usepackage{rotating}
\usepackage{xcolor}
\usepackage{xspace,color}
\usepackage{url}
\usepackage{listings}
\usepackage{multirow}
\usepackage{multicol}
\usepackage{verbatim}
\newtheorem{theorem}{Theorem}
\newtheorem{assump}{Assumption}

\newtheorem{lemma}[theorem]{Lemma}

\newtheorem{remark}[theorem]{Remark}

\newcommand{\red}[1]{{#1}}

\usepackage{array}
\newcolumntype{H}{>{\setbox0=\hbox\bgroup}c<{\egroup}@{}}

\usepackage{tikz,xcolor,hyperref}

\definecolor{lime}{HTML}{A6CE39}
\DeclareRobustCommand{\orcidicon}{%
	\begin{tikzpicture}
	\draw[lime, fill=lime] (0,0) 
	circle [radius=0.16] 
	node[white] {{\fontfamily{qag}\selectfont \tiny ID}};
	\draw[white, fill=white] (-0.0625,0.095) 
	circle [radius=0.007];
	\end{tikzpicture}
	\hspace{-2mm}
}

\foreach \x in {A, ..., Z}{%
	\expandafter\xdef\csname orcid\x\endcsname{\noexpand\href{https://orcid.org/\csname orcidauthor\x\endcsname}{\noexpand\orcidicon}}
}

\begin{document}

\title{Goodness-of-fit tests for generalized Poisson distributions}
\author{
A. Batsidis\orcidA{} $^{1}$, B. Miloševi\'c\orcidB{}
$^{2}$, M.D. Jim\'enez--Gamero\orcidM{}$^{3}$ \vspace*{5pt}\\ 
 $^{1}$ Department of Mathematics, University of Ioannina, 45110 Ioannina, Greece\\
 $^{2}$ Faculty of Mathematics,  University of Belgrade,
Belgrade, Serbia \\
$^{3}$  Departmento de Estad\'istica e Investigaci\'on Operativa, Universidad de Sevilla,
Seville, Spain} 

\author{
A. Batsidis \orcidA{} $^{1}$, B. Miloševi\'c $^{*}$\orcidB{}
$^{2}$, M.D. Jim\'enez--Gamero\orcidM{}$^{3}$ \vspace*{5pt}\\ 
 $^{1}$ Department of Mathematics, University of Ioannina, Ioannina, Greece\\
 $^{2}$ Faculty of Mathematics,  University of Belgrade,
Belgrade, Serbia \\
$^{3}$  Departmento de Estad\'istica e Investigaci\'on Operativa, Universidad de Sevilla,
Seville, Spain\\
$^{*}$ Corresponding address: bojana@matf.bg.ac.rs
\\
\\
The final version of this preprint will be published in the journal {\it Statistics}
}

\date{\today}
\maketitle

\begin{abstract}
This paper presents and examines computationally convenient goodness-of-fit tests for the family of generalized Poisson distributions, which encompasses notable distributions such as the Compound Poisson and the Katz distributions. 
 The tests are consistent against fixed alternatives and their null distribution  can be consistently approximated by a parametric bootstrap. The goodness of the bootstrap estimator and the power for finite sample sizes are numerically assessed through an extensive simulation experiment, including comparisons with other tests. In many cases, the  novel
tests either outperform or match the performance of existing ones. Real data applications are considered for illustrative purposes.
\end{abstract}


\section{Introduction}

Count data arises in many research fields such as medicine, actuarial sciences, biology, health, and economics. The Poisson distribution
has traditionally served as a reference for modeling that sort of data. However, it falls short in cases of overdispersion, underdispersion, or zero inflation (high {percentage of} zero values). 
To address  such 
limitations, the statistical literature presents various count distributions tailored to address these issues. This paper focuses on the generalised Poisson (GP) distributions, as introduced by Meintanis \cite{Meintanis2008}. These distributions generalize the Poisson model in the sense that the probability generating function (pgf) of each of these laws,  satisfies a specific differential equation. A wide range of distributions belongs to this family of distributions. Examples are the family of Katz distributions and the family of Compound Poisson distributions.
GP distributions offer a versatile framework for count data modeling,  enabling users to better handle overdispersion, underdispersion, and zero inflation. An important aspect of data analysis is to check the goodness-of-fit (gof) of a given data set with a probabilistic model. Therefore, it is relevant to provide appropriate procedures for testing gof of GP distributions.

To the best of our knowledge, the only tailor-made gof test of this family is the one proposed in Meintanis \cite{Meintanis2008}. This test is 
  an $L^2$ norm of the empirical counterpart of  the  aforementioned differential equation. Starting from this equation, this paper proposes  new gof tests for that family of distributions. Their associated test statistics are built
following the reasoning in Nakamura and P\'{e}rez-Abreu  \cite{Nakamura&PA1993a} for testing gof to the Poisson distribution (see also Novoa-Mu\~noz and  Jim\'enez-Gamero \cite{Novoa2}
for the bivariate Poisson distribution, Jim\'enez-Gamero and  Alba-Fern\'andez \cite{PT} for the
Poisson-Tweedie distribution, Batsidis et al. \cite{metrika2020} for the Bell distribution, Milo\v{s}evi\'{c} et al. \cite{Bojana2021} for the geometric distribution,  Batsidis and Lemonte  \cite{Batsidis&Lemonte} for the Neyman type A distribution). The new test statistics are
 a function of the coefficients of the polynomial in the equation that results when one replaces the pgf by the empirical probability generating function (epgf) in the differential equation, multiplied by positive constants, which serve as 
 weight terms.

The paper is organized as follows. {After} recalling the definition of the GP distributions,  Section \ref{GP} {introduces} the novel test statistics. Section \ref{seclim} contains results about their limiting properties. As the asymptotic null distribution  is not useful for approximating the null distribution, this section shows  the consistency of a parametric bootstrap estimator of the  null distribution. The consistency of the proposed tests against fixed alternatives is also  derived.
Section \ref{seccomp} is dedicated to the computation of the test statistics for the Katz and the Compound Poisson distributions. 
Section \ref{Numerical} summarizes the results of an extensive simulation experiment designed to assess the finite sample performance of the proposed tests, and to  compare them  with others.
The new tests are tried for several weighting schemes. From the simulations results it emerges that: the use of appropriate weights increases notably the power of the tests, and a rule for the choice of the weights  is given when testing for the Katz  and  the Poisson-Poisson distributions; although there is no test having the highest power against all considered alternatives,  in many cases, the recommended novel tests either outperform or match
the performance of existing ones. 
Section \ref{application} presents the application
of the gof tests to  several real data sets, and illustrates the choice of the weights in a practical setting. 
All technical proofs are deferred to the Appendix.


Before ending this section we introduce some notation: all limits in this paper
are taken when $n \rightarrow \infty$, where $n$ denotes the sample size;
$\stackrel{\mathcal{L}}{\longrightarrow}$ denotes convergence in distribution;
$\stackrel{{P}}{\longrightarrow}$ denotes convergence in probability;
$\stackrel{a.s.}{\longrightarrow}$ denotes the almost sure
convergence; $I(A)$ denotes the indicator function of the set $A$; $\mathbb{N}_0= \mathbb{N} \cup 0=\{0,1,2,...\}$; for each sequence $w=(w_0, w_1, w_2, \ldots)$ of positive numbers,   $l^2_w$ denotes the separable Hilbert space
$l^2_w=\{z=(z_0,z_1,z_2,\ldots), \, z_k \in \mathbb{R},\,  \sum _{k\geq0}z_{k}^2w_k<\infty\}$
with the  inner product $\langle z^{(1)},z^{(2)} \rangle_w=\sum _{k\geq0}z^{(1)}_{k}z^{(2)}_{k}w_{k}$;
let $\|\cdot\|_w$ stand for the associated norm; $P_{\lambda,\theta}$, $E_{\lambda,\theta}$ and $Cov_{\lambda,\theta}$ denote probability, expectation and covariance, respectively,
by assuming that the data come from a GP distribution
with parameters $\lambda>0$ and $\theta \in \Theta\subseteq \mathbb{R}^{p}$;
$P_*$, $E_{*}$ and $Cov_*$ denote the conditional probability law,
the conditional expectation and the  conditional covariance, respectively,
given  the data $X_1,\ldots,X_n$.

\section{The test statistics} \label{GP}
Following   Meintanis \cite{Meintanis2008},
a random variable $X$ taking values in $\mathbb{N}_0$, with $E(X)<\infty$ and  pgf $g(t)$, is said to belong to the family of GP distributions with parameters $\lambda>0$ and $\theta \in \Theta \subseteq \mathbb{R}^{p}$, if its pgf  satisfies the differential equation
\begin{equation}\label{differential}
G_{1}(t;\theta)g'(t)-\lambda G_{2}(t;\theta) g(t)=0, \text{ } g(1)=1, \quad \forall t \in [0,1],
\end{equation}
for a specific pair of functions $G_{1}(t;\theta)$ and $G_{2}(t;\theta)$, \red{where $g'(t)=\frac{{\rm d}}{{\rm d}t}g(t)$}. 

{\begin{remark}
At this point we must mention that the term generalized Poisson
distribution previously used should be no confused with the distributions
{appearing} in Section 9.3 of Johnson et al. \cite{UnivariateDiscreteDistributions}, which are also known as
Poisson-stopped sum or Compound Poisson distributions. In this paper, the term generalized Poisson distribution is in accordance with the term used by
Meintanis \cite{Meintanis2008} to describe the distributions which satisfy the differential
equation (\ref{differential}). As we will show in Section \ref{seccomp} examples of such distributions are
the family of Katz distributions, which contains the binomial  Poisson  and negative binomial as special cases, and the family of Compound Poisson distributions, which contains among others the Poisson, the negative binomial and the Hermite distribution. Since  the binomial distribution, which is obtained as a special case of the Katz distribution, is not a Compound Poisson distribution, the Katz family does not entirely lie within the Compound Poisson class (see Meintanis \cite{Meintanis2008}) and it is concluded that the two families cannot be described by the same stochastic representation.
\end{remark}}

Notice that the functions $G_{1}(t;\theta)$ and $G_{2}(t;\theta)$ are not uniquely defined. Moreover, since $g(t)$ and $g'(t)$ are positive in $(0,1]$ (if $P(X=0)>0$ and $P(X=1)>0$, then $g(t)$ and $g'(t)$ are positive in $[0,1]$), if {$G_{1}(t_0;\theta)=0$ ($G_{2}(t_0;\theta)=0$) for some $t_0\in (0,1]$, then we also have that $G_{2}(t_0;\theta)=0$ ($G_{1}(t_0;\theta)=0$).}

From now on, we will work with  GP distributions satisfying
$G_{1}(t;\theta) \neq 0$, $\forall t \in [0,1]$, which is equivalent to $G_{2}(t;\theta) \neq 0$, $\forall t \in [0,1]$. In such a case, instead of \eqref{differential}, we can write
\begin{equation}\label{differential.21}
g'(t)-\lambda G_{21}(t;\theta) g(t)=0, \text{ } g(1)=1, \quad \forall t \in [0,1],
\end{equation}
or
\[ 
G_{12}(t;\theta) g'(t)-\lambda  g(t)=0, \text{ } g(1)=1, \quad \forall t \in [0,1],
 \] 
where $G_{21}(t;\theta)=G_{2}(t;\theta) / G_{1}(t;\theta)$ and $G_{12}(t;\theta)=G_{1}(t;\theta) / G_{2}(t;\theta)$.
 Notice that, unlike  $G_{1}(t;\theta)$ and $G_{2}(t;\theta)$, $G_{12}(t;\theta)$ and $G_{21}(t;\theta)$ are uniquely defined. Moreover, since both $g(t)$ and $g'(t)$ are analytic functions, it readily follows that $G_{21}(t;\theta)$ must be an analytic function,
 \begin{equation}\label{Q.21}
 G_{21}(t;\theta)=\sum_{k \geq 0} q_k(\theta)t^k, \quad \forall t\in  [0,1],
\end{equation}
$q_0(\theta),  q_1(\theta), \ldots$ being real functions defined on $\Theta$. Specifically, 
$q_k(\theta)=G^{(k)}_{21}(0;\theta)/k!$, where  $G^{(k)}_{21}(0;\theta)$ denotes the $k$-th derivative of $G_{21}(t;\theta)$ with respect to $t$ evaluated at $t=0$. 

\red{ Let $p_k=P(X=k)$, $\forall k \geq 0$. Evaluating \eqref{differential.21} at $t=0$, one gets that
$$p_1=\lambda G_{21}(0;\theta) p_0=\lambda q_0(\theta) p_0,$$ 
and evaluating \eqref{differential.21} at $t=1$, one gets that
\[
E(X)=\lambda G_{21}(1;\theta)= \lambda\sum_{k \geq 0} q_k(\theta),
\]
which} implies that $\sum_{k \geq 0} q_k(\theta)<\infty$.
     Finally, as
$$g(t)=\sum_{k \geq 0} p_kt^k, \quad g'(t)=\sum_{k \geq 1} kp_kt^{k-1},$$
taking into account (\ref{differential.21}), after some simple algebra, the following characterizing relation holds,
\begin{equation}\label{relation}
(k+1)p_{k+1}=\lambda\sum_{u=0}^kp_uq_{k-u}(\theta), \quad \forall k \geq 0.
\end{equation}
Observe that \eqref{relation} implies that  $p_k>0$, $ \forall k \geq 0$.
\begin{remark}
Since  $G_{21}(t;\theta)\neq 0$, $\forall t \in [0,1]$, it follows that $G_{12}(t;\theta)=1/G_{21}(t;\theta)$ is also analytic,  and a similar reasoning applies.
\end{remark}

 From now on, all theoretical developments and discussions assume that the equation to deal with is that in  \eqref{differential.21} and, in short,  we will write $X\sim GPD(\lambda,\theta, G_{21}(\cdot;\theta))$.

\begin{remark} \label{charact}
Observe that if there exists a $g$ solving \eqref{differential.21}, then $g$ is unique and {is} given by $g(t)=c\exp(\lambda \red{\int_0^t G_{21}(s;\theta){\rm d} s)}$, with $c$ so that $g(1)=1$.
\end{remark}

Now we have all the ingredients to propose the novel class of goodness-of-fit tests of {the} composite null hypothesis
\[H_0:  X\sim GPD(\lambda,\theta,G_{21}(\cdot; \theta)), \text{ for some }\lambda>0, \, \, \theta \in \Theta, \]
against the alternative
\[H_1: X\nsim GPD(\lambda,\theta,G_{21}(\cdot; \theta)), \,  \forall \lambda>0, \, \, \forall \theta \in \Theta,\]
where, for each fixed value of $\theta$,  $G_{21}(\cdot; \theta)$ is a known real function.

Let  $X_1,\ldots, X_n$ be independent and identically distributed (iid) from a  random variable $X \in \mathbb{N}_{0}$, with $E(X)<\infty$ and pgf $g(t)$. Let
\begin{eqnarray*}
    g_n(t) & = & \frac{1}{n}\sum_{i=1}^{n}t^{X_i}=\sum_{k\geq 0}\hat{p}_kt^{k}, \\ 
    g_n'(t)& = & \frac{d}{dt}g_{n}(t)=\frac{1}{n}\sum_{i=1}^{n}X_it^{X_i-1} I(X_i\geq 1)=\sum_{k \geq 1}k\hat{p}_kt^{k-1},\\ 
\end{eqnarray*}
denote the epgf of $X_1,\ldots, X_n$ and its derivative, where 
\begin{equation} \label{prob}
\hat{p}_k=\frac{1}{n}\sum_{i=1}^{n}I(X_i=k), \quad k \in \mathbb{N}_0.
\end{equation}

Since the pgf of the  $GPD(\lambda,\theta,{G_{21}}(\cdot; \theta))$ distribution is the only pgf satisfying the diffe\-ren\-tial equation \eqref{differential.21} (see Remark \ref{charact}), and the fact that  the pgf $g(t)$ and its derivatives can be consistently estimated by the epfg and the derivatives of the epgf, respectively (see  Proposition 1  in Novoa-Mu\~noz and Jim\'enez-Gamero \cite{Novoa2014}), if $H_{0}$ is true, 
\[ 
D_{n}(t;\hat{\lambda},\hat{\theta})=g_{n}'(t)-\hat{\lambda}G_{21}(t; \hat{\theta})g_{n}(t),
\] 
should be close to 0, $\forall t \in [0,1]$, where   $\hat{\lambda}=\hat{\lambda}(X_1, \ldots, X_n)$  and $\hat{\theta}=\hat{\theta}(X_1, \ldots, X_n)$ are consistent estimators of $\lambda$ and $\theta$, respectively.

Following a similar line of reasoning to Nakamura and P\'{e}rez-Abreu \cite{Nakamura&PA1993a},
we express  $ D_{n}(t;\hat{\lambda},\hat{\theta})$  as
\[
D_n(t;\hat{\lambda},\hat{\theta})=\sum_{k\geq 0}\hat{d}(k; \hat{\lambda},\hat{\theta})t^k,
\]
with
\begin{equation}\label{dhat}
\hat{d}(k;\hat{\lambda},\hat{\theta})=
(k+1)\hat{p}_{k+1}-\lambda\sum_{u=0}^k \hat{p}_u q_{k-u}(\hat{\theta})=(k+1)\hat{p}_{k+1}-\lambda\sum_{u=0}^k \hat{p}_{k-u} q_{u}(\hat{\theta}), \quad \forall k \geq 0,
\end{equation}

and propose the family of test statistics
\begin{equation}\label{snthetalambda}
S_{n,w}(\hat{\lambda},\hat{\theta})=\sum_{k\geq 0}\hat{d}(k; \hat{\lambda},\hat{\theta})^2 w_{k}=
\| \hat{d}(\cdot;\hat{\lambda},\hat{\theta})\|_{w}^{2},
\end{equation}
where $\hat{d}(\cdot;\hat{\lambda},\hat{\theta})=(\hat{d}(0;\hat{\lambda},\hat{\theta}), \hat{d}(1; \hat{\lambda},\hat{\theta}), \ldots)$, and $w=(w_0,w_1,...)$ is such that 
\begin{equation} \label{w}
0 < w_k \leq M, \quad \forall k \geq 0,
\end{equation}
{$M$ being a positive constant.}
Section \ref{seclim} studies some asymptotic  properties of the proposed test statistics  
\eqref{snthetalambda}. Additionally, Section \ref{seccomp}  provides detailed expressions of the quantities 
$\hat{d}(k;\hat{\lambda},\hat{\theta})$ in \eqref{dhat} tailored to specific classes of distributions within the GP family.

\section{Limiting properties of the test statistics}\label{seclim}



In this section, we assume that {the } function ${G}_{21}(t; \theta)$ satisfies the following set of assumptions.

\begin{assump} \label{Q}

\begin{enumerate} \itemsep=0pt
\item[(i)] $\sum_{k \geq 0}  \left | q_k(\theta)\right | <\infty$.
\item[(ii)] $\frac{\partial}{ \partial \theta_i} q_k(\theta)$ exists $\forall k \geq 0$, and $\sum_{k \geq 0} \left | \frac{\partial}{ \partial \theta_i} q_k(\theta) \right|<\infty$, $ 1 \leq i \leq p$.
\item[(iii)] $\frac{\partial}{ \partial \theta_i} q_k(\theta)$ exists $\forall k \geq 0$, and there exists  an open  neighborhood $\mathcal{N}(\theta)$ of $\theta$, such that if $\theta_k \in \mathcal{N}(\theta)$, $\forall k \geq 0$, then $\sum_{k \geq 0} \left|  q_k(\theta_k)-q_k(\theta) \right|<\infty$, and  $\sum_{k \geq 0} \left |
    \frac{\partial}{ \partial \theta_i} q_k(\theta_k) - \frac{\partial}{ \partial \theta_i} q_k(\theta)\right |<\infty$, $ 1 \leq i \leq p$.
    
\item[(iv)] $\frac{\partial}{ \partial \theta_i} q_k(\theta)$ exists $\forall k \geq 0$, and { if, for each $k \geq 0$, $\left\{\theta_{kn} \right\}_{n\geq 1}$ is a sequence such that $\displaystyle \sup_{k \geq 0} \|\theta_{kn} -\theta\| \to 0$, } then
$\sum_{k \geq 0} \left|  q_k(\theta_{kn})-q_k(\theta) \right| \to 0$,  and $\sum_{k \geq 0} \left |\frac{\partial}{ \partial \theta_i} q_k(\theta_{kn}) - \frac{\partial}{ \partial \theta_i} q_k(\theta)\right | \to 0$, $ 1 \leq i \leq p$.
\item[(v)] If,  {  for each $k \geq 0$, $\left\{\theta_{kn} \right\}_{n\geq 1}$ is a sequence such that $\displaystyle \sup_{k \geq 0} \|\theta_{kn} -\theta\| \to 0$, } then $\sum_{k \geq 0} k {\left |  q_k(\theta_{kn})-q_k(\theta) \right | } \to 0$. 
\end{enumerate}
\end{assump}

\red{
Before stating properties of the test statistics, we first comment on Assumption 
\ref{Q}. It mainly ensures that $q_k(\theta)$ and its derivatives with respect to $\theta_j$ are continuous as functions of $\theta$. Implications in terms of the moments of $ X\sim GPD(\lambda,\theta,G_{21}(\cdot; \theta))$ are as follows.  Recall that if $X$ is a random variable taking values in  $\mathbb{N}_{0}$ with pgf $g$, then $E(X)=g'(1)$ and 
 $E(X^2)=g''(1)+g'(1)$,
where 
$g''(t)=\frac{{\rm d}^2}{{\rm d} t^2}g(t)$. Now, if $ X\sim GPD(\lambda,\theta,G_{21}(\cdot; \theta))$, then
\begin{itemize}
    \item $E_{\lambda, \theta}(X)=\lambda \sum_{k \geq 0}q_k(\theta)$.  If Assumption \ref{Q} (i) holds then  $E_{\lambda, \theta}(X)<\infty$, as it was observed in Section   \ref{GP};  moreover,  if Assumption 
\ref{Q} (iv) also holds and $(\lambda_n, \theta_n) \to (\lambda, \theta)$, then 
 $E_{\lambda_n, \theta_n}(X) \to E_{\lambda, \theta}(X)$, that is, Assumption 
\ref{Q} (iv) ensures that $E_{\lambda, \theta}(X)$ is continuous as a function of 
 $(\lambda, \theta)$.
    \item $E_{\lambda, \theta}(X^2)=\lambda \sum_{k \geq 0}q_k(\theta)+\left(\lambda \sum_{k \geq 0}q_k(\theta)\right)^2+\lambda \sum_{k \geq 0}kq_k(\theta)$. If  $E_{\lambda, \theta}(X^2)<\infty$, Assumption 
\ref{Q} (v) holds, and $(\lambda_n, \theta_n) \to (\lambda, \theta)$, then 
 $E_{\lambda_n, \theta_n}(X^2) \to E_{\lambda, \theta}(X^2)$, that is, Assumption 
\ref{Q} (v) ensures that $E_{\lambda, \theta}(X^2)$ is continuous as a function of 
 $(\lambda, \theta)$.
    \end{itemize}
}

\begin{theorem}\label{power1}
Let  $X_1,\ldots, X_n$ be iid from a  random variable $X \in \mathbb{N}_{0}$, with $E(X^2)<\infty$.
Assume that {$\hat{\lambda} \stackrel{a.s.(P)}{\longrightarrow} \lambda$}, $\hat{\theta} \stackrel{a.s.(P)}{\longrightarrow} \theta$, \red{ $(\lambda,\theta)$ being the true parameter values if $H_0$ is true}, and that Assumption \ref{Q} (i)-(iii) hold.
Then
$$S_{n,w}(\hat{\lambda},\hat{\theta} )\stackrel{a.s.(P)}{\longrightarrow} \eta={\| {d}( \cdot ;\lambda, \theta) \|_w^2},$$
where $d(\cdot;\lambda,\theta)=(d(0;\lambda,\theta),d(1; \lambda,\theta),\ldots)$, and
$d(k;\lambda,\theta)=(k+1){p}_{k+1} -{\lambda}\sum_{u=0}^{k}q_{k-u}({\theta}){p}_{u}$, $k\in \mathbb{N}_{0}$.
\end{theorem}

\red{Theorem \ref{power1} assumes that the limit of 
$\hat{\lambda}$ and $\hat{\theta}$  are the true parameter values when $H_0$ is true. An interpretation of the meaning of $(\lambda,\theta)$ when $H_0$ is not true is given after Theorem \ref{boot}. }

Notice that $\eta\geq 0$, with $\eta=0$
if and only if $H_{0}$ is true. Hence, the null hypothesis $H_{0}$ should be
rejected for large values of the test statistic $S_{n,w}( \hat{\lambda},\hat{\theta} )$.
Now, to determine what is a large value we have to obtain the
distribution of the test statistic $S_{n,w}(\hat{\lambda},\hat{\theta})$ under
the null hypothesis $H_{0}$, or at least an
approximation to it. With this aim, we next derive its
asymptotic null distribution. It will be assumed that
the estimators $\hat{\lambda}$ and $\hat{\theta}$
satisfy the following regularity condition.

\begin{assump} \label{estim}
   $H_{0}$ is true,   $\theta\in \Theta$ and $\lambda>0$ are the true parameter values, and  
\begin{eqnarray*}
\sqrt{n} (\hat{\lambda}-\lambda) & = & \frac{1}{\sqrt{n}}\sum_{i=1}^n \psi_0(X_i; \lambda, \theta) +r_0,\\
\sqrt{n} (\hat{\theta}-\theta) & = & \frac{1}{\sqrt{n}}\sum_{i=1}^n \psi(X_i;  \lambda, \theta) + r,
\end{eqnarray*}
$\psi(x;  \lambda, \theta)^\top=(\psi_1(x; \lambda, \theta), \ldots, \psi_p(x; \lambda, \theta))$,
with $E_{\lambda,\theta}\{\psi_j(X;  \lambda, \theta)\}=0$, $E_{\lambda,\theta}\{\psi_j(X;  \lambda, \theta)^2\}<\infty$, $0\leq j \leq p$, $r^\top=(r_1, \ldots, r_p)$, and
 $r_{j} \stackrel{P}{\longrightarrow} 0,$  $0 \leq j \leq p$.
\end{assump}

Assumption \ref{estim} implies that 
when the null hypothesis is true and
$\lambda$ and $\theta$ denote the true parameter values, 
then $\sqrt{n} (\hat{\lambda}-\lambda)$ and $\sqrt{n} (\hat{\theta}-\theta)$
are asymptotically normally distributed. This assumption is not restrictive at all since it is
fulfilled by commonly used estimators such as the  maximum likelihood estimators and  moment estimators, as well as by other estimators such as those in \cite{JG&Batsidis2017,AISM2016}, obtained {by} minimizing certain distances. \red{This assumption is commonly supposed to hold in gof testing problems with composite null hypothesis (see, e.g. Assumption (R1) in \cite{Klar}, Assumption 1 in \cite{Meintanis2008} or Assumption A1 in \cite{Meintanis&S}). }

The next theorem gives the asymptotic null distribution of $nS_{n,w}(\hat{\lambda},\hat{\theta})$.

\begin{theorem} \label{asymptoticnulldistribution}
Let $X_1,  \ldots, X_n$ be iid from
$X\sim {GPD}(\lambda,\theta,G_{21}(\cdot; \theta))$, with $E_{\lambda,\theta}(X^2)<\infty$, for some $\lambda>0$ and some  $\theta \in \Theta$.
Assume that $w$ satisfies \eqref{w},  that  Assumption \ref{Q} (i), (ii) and (iv) hold, and that $\hat{\lambda}$ and $\hat{\theta}$ satisfy Assumption  \ref{estim}. Then,
${n}S_{n,w}(\hat{\lambda},\hat{\theta})\stackrel{\mathcal{L}}{\longrightarrow} \|G(\lambda,\theta)\|_w^2$,
where $\{G(\lambda,\theta)=(G(0;\lambda,\theta), G(1;\lambda,\theta), \ldots)\}$ is a centered Gaussian
process in $l^2_w$ with covariance kernel
$C(k,s)=Cov_{\lambda,\theta}\{Y(X;k,\lambda,\theta),Y(X;s, \lambda,\theta)\}$, $k,s \in \mathbb{N}_{0}$, 
\begin{eqnarray}
Y(X;k,\lambda,\theta) & = & \phi(X;k,\lambda,\theta)+ \sum_{j=0}^p \mu_j(k;\lambda,\theta) \psi_j(X;  \lambda, \theta), \label{Y.expression} \\
   \phi(x;k,\lambda,\theta)& = & (k+1)I(x=k+1) -{\lambda}\sum_{u=0}^{k}q_{k-u}({\theta})I(x=u), \label{phik}
 \end{eqnarray}
$  \psi_0(X;  \lambda, \theta), \ldots,  \psi_p(X;  \lambda, \theta)$ are as defined in Assumption \ref{estim}, and
\begin{eqnarray*}
\mu_0(k;\lambda,\theta) & = &   E_{\lambda,\theta}\left\{\frac{\partial}{\partial \lambda}\phi(X;k,\lambda,\theta)\right\}=-\sum_{u=0}^{k}q_{k-u}({\theta})p_u,\\
\mu_j(k;\lambda,\theta)  & = &  E_{\lambda,\theta}
\left\{\frac{\partial}{\partial \theta_j}\phi(X;k,\lambda,\theta)\right\}=-\lambda\sum_{u=0}^{k}\frac{\partial}{\partial \theta_j}q_{k-u}({\theta})p_u, \quad 1 \leq j \leq p.
\end{eqnarray*}
\end{theorem}

The distribution of $\|G(\lambda,\theta)\|_w^2$  is the same as that of $\sum_{j=1}^\infty \lambda_jN_j^2$, where $\lambda_1, \lambda_2, \ldots $ are the
positive eigen\-va\-lues of the integral operator $f \mapsto Af$  on $\ell^2_w$ associated with the covariance kernel $C(\cdot,\cdot)$ given in Theorem \ref{asymptoticnulldistribution},
i.e., $(Af)(k) \! = \! \sum_{s\geq 0} \! C(k,s) f(s) w_s$,
and $N_1,N_2, \ldots $ are iid standard normal random variables. \red{In ge\-ne\-ral, calculating the eigenvalues associated with the covariance kernel of a test statistic is a hard problem. Moreover, solutions for testing composite null hypotheses are very scarce. Examples are Stephens \cite{Stephens1976}, for testing gof to the univariate normal and exponential distributions with tests based on the empirical distribution function{,} and Ebner and Henze \cite{ebner2023eigenvalues}, for testing gof to the univariate normal distribution with a test based on the empirical characteristic function. A general procedure for consistently approximating the eigen\-va\-lues is described  in  Section 3.2 of Novoa-Mu\~noz and Jim\'enez-Gamero \cite{Novoa2}, that was applied to the problem of testing gof to the bivariate Poisson distribution. Its application requires the use of consistent estimators of the functions $\psi_0, \psi_1, \ldots, \psi_p $  in Assumption \ref{estim}.
} {For some alternative approximation methods we also refer to \cite{bovzin2020new} (see also \cite{meintanis2023bahadur}), \cite{ebner2023bahadur}, \cite{ebner2024eigenvalues} and the references cited therein.
}

\red{To approximate the null distribution of the proposed test statistic, next we study a parametric bootstrap,}
defined as follows: let
 $ {X}^*_{1}$,
$\ldots$, ${X}^*_{n}$ be iid from  $X^*\sim {GPD}(\hat{\lambda},\hat{\theta},G_{21}(\cdot; \hat{\theta}))$, and
let $S_{n,w}^*(\hat{\lambda}^*,\hat{\theta}^*)$ be the bootstrap version
of $S_{n,w}(\hat{\lambda},\hat{\theta})$ obtained by replacing ${X}_{1},  \ldots,  {X}_{n}$, $\hat{\lambda}  =  \hat{\lambda} ({X}_{ 1},  \ldots,  {X}_{n})$ and
$\hat{\theta}  =  \hat{\theta} ({X}_{ 1},  \ldots,  {X}_{n})$ with $ {X}^*_{ 1},  \ldots,  {X}^*_{n}$, $\hat{\lambda}^{*}= \hat{\lambda}({X}^*_{1}, \ldots, {X}^*_{n})$ and
$\hat{\theta}^{*}= \hat{\theta}({X}^*_{1}, \ldots, {X}^*_{n})$, respectively,
in the expression of $S_{n,w}(\hat{\lambda},\hat{\theta})$. Then, the parametric bootstrap approximates $P_{\lambda , \theta}\{S_{n,w}(\hat{\lambda},\hat{\theta} )\leq x\}$
by means of its bootstrap version, i.e.~$P_*\{S_{n,w}^*(\hat{\lambda}^*,\hat{\theta}^*)\leq x \}=P_{\hat{\lambda} , \hat{\theta} }\{S_{n,w}^*(\hat{\lambda}^*,\hat{\theta}^*)\leq x \,|\, X_1, \ldots, X_n\}$. The consistency of this  estimator is stated in  Theorem \ref{boot}, which assumes that  Assumption \ref{estim2} below, holds true.

\begin{assump} \label{estim2} Let $\{\lambda_n\}_{n \geq 1}$  and $\{\theta_n\}_{n \geq 1}$ be two sequences such that   $\lambda_n \to \lambda>0$ and 
$\theta_n \to \theta \in \Theta$, then

\begin{enumerate}
\item[(i)] {Assumption \ref{estim} holds with $\theta$ and $\lambda$ replaced with  $\theta_n$ and $\lambda_n$, that is, if $X_{n1}, \ldots, X_{nn}$ are iid with {$X_{ni} \sim {GPD}(\lambda_n,\theta_n,G_{21}(\cdot; \theta_n))$}, $1\leq i \leq n$, $\hat{\lambda}_n=\lambda(X_{n1}, \ldots, X_{nn})$ and 
$\hat{\theta}_n=\theta(X_{n1}, \ldots, X_{nn})$,  
then
\begin{eqnarray*}
\sqrt{n} (\hat{\lambda}_n-\lambda_n) & = & \frac{1}{\sqrt{n}}\sum_{i=1}^n \psi_0(X_{ni}; \lambda_n, \theta_n) +r_{n0},\\
\sqrt{n} (\hat{\theta}_n-\theta_n) & = & \frac{1}{\sqrt{n}}\sum_{i=1}^n \psi(X_{ni};  \lambda_n, \theta_n) + r_n,
\end{eqnarray*}
where $\psi(x;  \lambda, \theta)$ is as defined in  Assumption \ref{estim},
$E_{\lambda_n,\theta_n}\{\psi_j(X;  \lambda_n, \theta_n)\}=0$, $E_{\lambda_n,\theta_n}\{\psi_j(X;  \lambda_n, \theta_n)^2\}<\infty$, $0\leq j \leq p$, $r_n^\top=(r_{n1}, \ldots, r_{np})$, and
 $P_{\lambda_n,\theta_n}(|r_{nj}|>\varepsilon) \to 0$, $\forall \varepsilon>0,$  $0 \leq j \leq p$.}

\item[(ii)] $E_{\lambda_n,\theta_n}\{\psi_j(X;\lambda_n,\theta_n)\psi_v(X;\lambda_n,\theta_n)\} \to E_{\lambda,\theta}\{\psi_j(X; \lambda,\theta) \psi_v(X;\lambda,\theta)\}<\infty$,  $0 \leq j,v \leq p$.
\item[(iii)]  $\psi_j(k;\lambda_n,\theta_n)\to \psi_j(k;\lambda,\theta) $, $\forall k \in \mathbb{N}_0$,  $0 \leq j \leq p$.
\item[(iv)]  $E_{\lambda_n,\theta_n}\{X\psi_j(X;\lambda_n,\theta_n)\} \to E_{\lambda,\theta}\{X\psi_j(X; \lambda,\theta)\} $,  $0 \leq j \leq p$.
\item[(v)] $ \displaystyle \sup_{\vartheta \in \mathcal{N}}E_\vartheta\left(\|\zeta(X; \vartheta)\|^2 I(\|\zeta(X;\vartheta)\|>\varepsilon\sqrt{n}\right)\to 0$, $\forall \varepsilon>{0}$, where  $\mathcal{N}$ is an open neighborhood of $(\lambda, \theta_1, \ldots, \theta_p)^\top$, and $\zeta(x; \vartheta)^\top=(\psi_0(x; \vartheta), \psi_1(x; \vartheta), \ldots, \psi_p(x; \vartheta))$.
\end{enumerate}
\end{assump}

Recall that Assumption \ref{Q} implies that the first {and} second moments of a random variable having a GP distribution are continuous functions when considered as functions of the parameters. Now, to prove the consistency of the parametric bootstrap approximation to the null distribution of the test statistic, stronger assumptions are needed to ensure that certain quantities are also  continuous functions of the parameters, as stated in Assumption \ref{estim2}. Most parts of Assumption \ref{estim2} are commonly assumed when proving the consistency of a parametric bootstrap distribution estimator. For example,  Assumption \ref{estim2} (i), (ii) and (v) are Assumption P, (B.1)  and (B.4), respectively in \cite{Meintanis&S}; Assumption \ref{estim2} (i), (ii), and (iii) is Assumption (R.1) in \cite{Klar} and Assumption A.1 in \cite{Henze1996}.

\begin{theorem}\label{boot}
Let $X_1, \ldots, X_n$ be iid from $X$, a random variable
taking values in $\mathbb{N}_{0}$. Assume that $\hat{\lambda} \stackrel{a.s. (P)}{\longrightarrow} \lambda$,
$\hat{\theta} \stackrel{a.s.(P)}{\longrightarrow} \theta$,
for some $\lambda>0$ and some $\theta \in \Theta$,  so that $E_{\lambda,\theta}(X^2)<\infty$ \red{and $(\lambda,\theta)$ are the true parameter values if $H_0$ is true}, that
$w=(w_0,w_1,...)$ satisfies \eqref{w},  that  Assumptions \ref{Q} and  \ref{estim2} hold. Then,
$$ \sup_{x\in\mathbb{R}} \left|P_*\{nS_{n,w}^*(\hat{\lambda}^*,\hat{\theta}^*) \leq x \}
-P_{\lambda,\theta}\{nS_{n,w}(\hat{\lambda},\hat{\theta})\leq x\}\right| \stackrel{a.s.(P)}{\longrightarrow} 0.$$
\end{theorem}

\red{Theorem \ref{boot} holds whether $H_0$ is true or not. It states that
the conditional distribution of $S_{n,w}^*(\hat{\lambda}^*,\hat{\theta}^*)$
and the distribution of $S_{n,w}(\hat{\lambda},\hat{\theta})$ are close when the sample
is drawn from $X\sim {GPD}(\lambda,\theta,G_{21}(\cdot; \theta))$, 
$\lambda$ and $\theta$ being the  limits of $\hat{\lambda}$ and $\hat{\theta}$, respectively. In particular,
if the null hypothesis $H_0$ is true, then Theorem \ref{boot} states that the conditional
distribution of $S_{n,w}^*(\hat{\lambda}^*,\hat{\theta}^*)$ is close to the
null distribution of $S_{n,w}(\hat{\lambda},\hat{\theta})$. If the null hypothesis $H_0$ is not true, then ${GPD}(\lambda,\theta,G_{21}(\cdot; \theta))$ can be interpreted as the projection of the true law of the data, $L$, on the set ${\cal G}=\{ {GPD}(\lambda,\theta,G_{21}(\cdot; \theta)), \, \lambda>0, \, \theta \in \Theta\}$, that is, it is  the point in  ${\cal G}$ closest to $ L$. The precise meaning of closeness depends on the method used to derive the parameter estimators $\hat{\lambda}$ and $\hat{\theta}$. For example, if maximum likelihood estimators are employed, then  ${GPD}(\lambda,\theta,G_{21}(\cdot; \theta))$ minimizes the Kullback-Leibler distance between $L$ and ${\cal G}$ (see \cite{white}); the estimators can be also chosen so that they minimize a distance between the pfg of the data and the pfg of a law in ${\cal G}$, and in this case the meaning of closeness is clear (see \cite{JG&Batsidis2017}). 
The existence and uniqueness of $\lambda>0$ and $\theta \in \Theta$ minimizing a distance (or, in general, a  dissimilarity) between $ L$ and ${\cal G}$ is not guaranteed in general (see, e.g. Example 2 in 
\cite{Jim2014})}.

Let $\alpha \in (0,1)$. As a consequence of Theorem  \ref{boot}  the test function 
\[
\Psi^*=\left\{
         \begin{array}{ll}
           1, & \text{if}\ S_{n,w}(\hat{\lambda},\hat{\theta})\geq s^*_{n,\alpha}, \\
           0, & \text{otherwise},
         \end{array}
       \right.
\]
{is asymptotically correct in the sense that when $H_0$ is true, its level goes to $\alpha$,
 where  $s^*_{n,\alpha}=\inf\{x:P_*(S^*_{n,w} (\hat{\lambda}^{*},\hat{\theta}^*)\geq x)\leq \alpha\}$ is the $\alpha$
upper percentile of the bootstrap distribution of $S_{n,w} (\hat{\lambda},\hat{\theta})$.}

An immediate consequence of Theorem \ref{power1} and Theorem \ref{boot}
is that, \red{if the assumptions in those theorems hold,} then the test $\Psi^*$ is consistent; that is,
it is able to detect any fixed alternative, in the sense that
$\Pr(\Psi^*=1) \to 1$ whenever $X\nsim {GPD}(\lambda,\theta, G_{21}(\cdot; \theta))$, for any $\lambda, \theta$.

Due to computational reasons, as usual,  the bootstrap distribution of the test statistic is approximated by simulation generating  $B$ bootstrap samples (see e.g. \cite{Batsidis&Lemonte}). 


\section{Computation of the test statistic in some specific GP distributions }\label{seccomp}
This section derives expressions for the test statistics tailored to the Katz and the Compound Poisson distributions. 

In the sequel, we denote  $M_1=\displaystyle \max_{1\leq i \leq n}X_i$.

\textbf{The Katz family of distributions.}
The law of $X$ belongs to the Katz family of distributions  with parameters $\lambda>0$ and $\theta<1$, 
if its pgf is given by
\[ 
g_{\text{Katz}}(t;\lambda,\theta)=\left(1-\theta\right)^{\frac{\lambda}{\theta}}\left(1-\theta t\right)^{-\frac{\lambda}{\theta}}.
\] 
Then we write that $X\sim \text{Katz}(\lambda,\theta)$. Note that $g_{\text{Katz}}(t;\lambda,\theta)$ satisfies relation {\eqref{differential}}  with
\[ 
G_{1}(t;\theta)=1-\theta t \text { and } G_{2}(t;\theta)=1,
\] 
or relation {\eqref{differential.21}} with 
\[ 
G_{21}(t;\theta)=(1-\theta t)^{-1}.
\] 
Thus, $\text{Katz}(\lambda,\theta)\equiv GPD\left(\lambda, \theta,(1-\theta t)^{-1}\right)$.
The Katz family reduces to the binomial
distribution when $\theta<0$ and $\theta (\theta-1)^{-1}=p$, $-\lambda/\theta=n$, to the Poisson distribution
when $\theta \rightarrow  0$  and to the negative binomial model when $0 < \theta = p < 1$, $\lambda/\theta= k$ (see Consul and Famoye \cite{Consul&Famoye} and references therein).

In this special case, 
\begin{eqnarray*}
 \hat{d}\big( 0;\hat{\lambda},\hat{\theta}\big) & = & \hat{p}_{1}-\hat{\lambda}\hat{p}_{0},\\
 \hat{d}\big(k;\hat{\lambda},\hat{\theta}\big) & = & (k+1) \hat{p}_{k+1}-\hat{\lambda}\sum_{u=0}^{k}\hat{p}_{u}\hat{\theta}^{k-u}, \quad k=1,...,M_1-1,\\
\hat{d}\big(k;\hat{\lambda},\hat{\theta}\big) & = & -
 \hat{\lambda} \sum_{u=0}^{M_1}\hat{p}_{u}\hat{\theta}^{k-u}=-\hat{\lambda}\sum_{u=k-M_1}^{k}\hat{p}_{k-u}\hat{\theta}^{u}, \quad k\geq M_1.
 \end{eqnarray*}

\textbf{Compound Poisson distributions.}
 The law of $X$ belongs to the family of Compound Poisson (CP) distributions with parameters $\lambda>0$ and $\theta$ if it is distributed as  $\sum_{j=1}^{N}Y_j$, where $Y_1, Y_2, \ldots $ are iid from $Y$  with pgf $g(t;\theta)$, $N$ is independent of $Y_1, Y_2, \ldots $ and follows a Poisson law with parameter $\lambda>0$. In such a case, we write that $X\sim CP(\lambda,\theta)$. The pgf of the CP distribution with parameters  $\lambda>0$  and $\theta$ is given by
 \[ 
 g_{CP}(t;\lambda,\theta)=\exp\left\{-\lambda [1-g(t;\theta)]\right\}.
\] 
 It is easily to verify that $g_{CP}(t;\lambda,\theta)$ satisfies relation {\eqref{differential}} for
 \[ 
G_{1}(t;\theta)=1 \text { and } G_{2}(t;\theta)=g'(t;\theta),
\] 
 where $g'(t;\theta)$ is the derivative with respect to $t$ of  $g(t;\theta)$.   Equivalently, $g_{CP}(t;\lambda,\theta)$ satisfies relation {\eqref{differential.21}} with 
\begin{equation}\label{g21CP}
G_{21}(t;\theta)=g'(t;\theta).
\end{equation}
Thus, $CP(\lambda,\theta)\equiv GPD\left(\lambda, \theta,{g'(t;\theta)}\right)$.
Taking into account relation \eqref{g21CP}, we have that
\[
G_{21}(t;\theta)=\sum_{k\geq 0}(k+1) P_{\theta}\left(Y=k+1\right)t^{k},
\]
and thus
$
q_k(\theta)=(k+1)P_{\theta}\left(Y=k+1\right)$.  As it is implicitly assumed that the law of $Y$ depends on $\theta$,
we have denoted $P_{\theta}$ to the  probability  {of} that random variable. 
Therefore
\begin{eqnarray*}
 \hat{d}\big(k;\hat{\lambda},\hat{\theta}\big) & = & (k+1)\hat{p}(k+1)-\hat{\lambda}\sum_{u=0}^{k}q_u(\hat{\theta})\hat{p}_{k-u}, \quad k=0,1,...,M_1-1,\\
 \hat{d}\big(k;\hat{\lambda},\hat{\theta}\big) & = & -\hat{\lambda}\sum_{u=0}^{M_1}q_{k-u}(\hat{\theta})\hat{p}_{u}=-\hat{\lambda}\sum_{u=0}^{k}q_u(\hat{\theta})\hat{p}_{k-u}=-\hat{\lambda}\sum_{u=k-M_1}^{k}q_u(\hat{\theta})\hat{p}_{k-u}, \quad k\geq M_1.
 \end{eqnarray*}

For example, if $Y$ has a Poisson distribution with parameter $\theta$, in which case the law of $X$ is called Poisson-Poisson and denoted as $PP(\lambda, \theta)\equiv GPD\left(\lambda, \theta,{\theta e^{\theta (t-1)}}\right)$,   we have $q_k(\theta)=e^{-\theta}{\theta^{k+1}}/{k!}$;   if $Y$ has a Binomial 
 distribution  with parameters $\theta=(\nu,p)$, in which case the law of $X$ is called Poisson-Binomial and denoted as $PB(\lambda,\nu,p) \equiv GPD\left(\lambda, (\nu,p),
\nu p (pt+q)^{\nu-1} \right)$,   we have $q_k(\theta)= {\nu-1\choose k}\nu p^{k+1}(1-p)^{\nu-1-k}$, $0\leq k \leq \nu-1$, $q_k(\theta)=0$, $k \geq \nu$.

\section{Simulation study}\label{Numerical}
This section summarizes the results of a simulation study. All results were obtained at a  significance level of $\alpha=0.05$, were approximated using $N=1000$ Monte Carlo replicates and, at each replicate, $B=750$ bootstrap cycles were used to calculate the bootstrap distribution. All computations were performed using the R language \cite{R}.
The code utilized for 
the calculations 
is available from the authors upon request.


For the proposed test statistic, seven weighting schemes were considered: all weight terms are  equal to 1 (in fact this case applies no weight), and the weight terms correspond to the probability mass function (pmf) of the Negative Binomial distribution,  $NB(\nu,p)$, with parameters $(\nu,p)=$ (2,0.25), (2,0.5), (2,0.75), (4,0.25), (4,0.5), (4,0.75).  The statistics are denoted by $S_{n,i},\; i=1,2...,7$, respectively. Recall that $NB(\nu,p)$ is the discrete probability distribution that models the number of failures in a sequence of iid Bernoulli trials with success probability $p$ before $\nu$  successes occurs. 

 As competitors of the 
 novel gof tests based on the statistics $S_{n,i},\; i=1,2...,7,$ we considered:
\begin{itemize}
    \item the test proposed in \cite{Meintanis2008}, based on the statistic $T_{n,\gamma}$, with $n$ the sample size and $\gamma>1$ a tuning parameter, that we took $\gamma=4$, following the recommendations in  \cite{Meintanis2008}; 
    \item the test based on the classical Anderson-Darling $AD_n$  and Cramér–von Mises $C_n$ test statistics  that were  studied in \cite{Henze1996}; 
    \item the test based on the statistic $R_n$ proposed in \cite{Rueda}; and
    \item the test based on the statistic $W_n$ proposed in \cite{Klar}.
\end{itemize}
{It should be emphasized that the test statistics $AD_n$, $C_n$ and $R_n$ are defined by means of an infinite sum and hence these sums have to be truncated to some finite value. To be more specific, for the $AD_n$ test statistic defined by
\begin{align*}
     AD_n=n\sum_{j=1}^{\infty}\frac{(F_n(j)-F(j;\hat{\lambda},\hat{\theta}))^2f(j;\hat{\lambda},{\theta})}{F(j;\hat{\lambda},\hat{\theta})(1-F(j;\hat{\lambda},\hat{\theta}))},
 \end{align*}
 when $M_1<30$ we replaced $\infty$ by $\min(30,k_0)$ where $k_0$ is the minimal number for which the denumerator in the expression above is equal or very close to zero, while if $M_1>30$ then we use $\min(30,k_0)$ summands. Finally, for the $C_n$ and $R_n$ test statistic we use $\max(30,M_1)$ summands.
}

To  estimate the parameters, moment estimators were used. In all cases, the null distribution was approximated through a parametric bootstrap.

 Tables \ref{Table01} and \ref{Table02} display empirical size and power results  with percentages of rejection rounded to the nearest integer, corresponding to the Poisson-Poisson null hypothesis. Similarly,  Tables \ref{Table03} and \ref{Table04} show corresponding results for the Katz null hypothesis. 
In each row of the power results, to facilitate comparison of the tests, we have  highlighted in boldface the best  competitor test  and the best among the proposed tests. First, it is worth noting that, for $n=50$, the novel tests are more conservative than competitors, while for  $n=100$ all tests have sizes close to the nominal value.


To evaluate the power of the tests, we generated samples from various alternative distributions: 
the beta mixture of binomial distributions $BB(v,a)$, where $X|p\sim B(v,p)$ and $p$ follows a beta distribution with both parameters equal to $a$;  the discrete uniform distribution $DU(\nu)$, defined in the range ${0, 1, \ldots, \nu}$; the mixture of a Poisson and a discrete uniform distribution $MPDU(\nu,\epsilon)$ (an observation is generated from a Poisson distribution with parameter equal to one with probability $\epsilon$, or from a discrete uniform distribution $DU(\nu)$ in the range ${0, 1, \ldots, \nu}$ with probability $1-\epsilon$); the mixture of a Poisson-Binomial and a discrete uniform distribution, denoted as $MPBDU(\lambda,v,p,\nu,\epsilon)$; 
the Poisson-Binomial distribution $PB(\lambda,v,p)$; the Negative Binomial distribution $NB(v,p)$; the mixture of a Katz and a discrete uniform distribution, denoted as $MKDU(\lambda,\theta,\nu,\epsilon)$;   
 the mixture of a Katz and a Poisson distribution, denoted as $MKP(\lambda,\theta,\nu,\epsilon)$;  
the mixture of a Negative Binomial and a Poisson distribution, denoted as $MNBP(\lambda,p,\nu,\epsilon)$; 
and $\max\{X_1,X_2\}$, where $X_1$ follows a Katz distribution $\text{Katz}(\lambda,\theta)$, and $X_2$ follows a $DU(\nu)$ distribution, denoted as $MaxKDU(\lambda,\theta,\nu)$.

\begin{table}
	\centering
	\footnotesize
	\caption{{Testing Poisson-Poisson distribution: empirical powers (percentage of rejection rounded to the nearest integer, $n=50$, $\alpha=0.05$).}}
	\label{Table01}
	\begin{tabular}{|c|ccccc|ccccccc|}
		\hline
		Alternative &$C_n$&$AD_n$ & $T_{n,4}$ & $R_n$ & $W_n$ & $S_{n,1}$  &$S_{n,2}$& $S_{n,3}$& $S_{n,4}$ &$S_{n,5}$& $S_{n,6}$ & $S_{n,7}$ \\
		\hline
  $PP(1,2)$&6&6 &5 &5 &7 & 3  & 4 &4 &5 & 3 &4 &5
  \\ 
  $PP(1,1.5)$&6&5 &4 &6 &4 &3  &3 &4 &5 &1 &3 &3   \\ 
  $PP(1,1)$&5&5 &4 &5 &5 &3 &3 &4 &6 &1 &3 &4  \\ \hline
	$BB(7,1)$&59 &\textbf{62} &63 &37 &{48} &46  &38 &14 &14 &\textbf{47} &32 &7   \\ 
		$DU(15)$&59 &\textbf{70} &47 &26  &\textbf{70} &30 &16 &16 &10&\textbf{32} &5 &17 \\ 
		$MPDU(10,0.25)$&31 &36 &13 &6 & \textbf{41} &43  &39 &17 &13 &\textbf{48} &25 &13  \\ 
         $MPDU(10,0.5)$&27 &\textbf{30} &2 &12 &13 &\textbf{40} &\textbf{40} &29 &{37} &{39} &30 &26\\
$MPBDU(1,3,0.75,3,0.25)$&\textbf{33}&30 &19 &13 &25 &59  &61 &57 &38 &46 &\textbf{62} &54  \\ 
$MPBDU(2,3,0.75,3,0.25)$&48& 48 &41 &37 &\textbf{50} & 54 &61 &\textbf{65} &48&33 &64 &62  \\ 
  $PB(1,3,0.75)$&\textbf{23}&\textbf{23} &12 &20 &15 & 14 &16 &27 &\textbf{31} &6 &17 &28  \\ 
$PB(2,3,0.75)$&5&8 &9 &\textbf{13} &4 & 4  &6 &10 &\textbf{15} &4 &7 &10  \\ 
$NB(2,0.5)$&\textbf{11}&10 &10&\textbf{11} &9 &4 &4&6&\textbf{9}&3&4&6\\ 
$NB(3,0.25)$&15& 19 &20 &16 &\textbf{22} &4 &8 &\textbf{15} &14 &4 &11 &\textbf{15}  \\ 
		\hline
	\end{tabular}
\end{table}

\begin{table}
	\centering
	\footnotesize
	\caption{{Testing Poisson-Poisson distribution: empirical powers (percentage of rejection rounded to the nearest integer $n=100,\alpha=0.05$).}}
	\label{Table02}
	\begin{tabular}{|c|ccccc|ccccccc|}
		\hline
		Alternative &$C_n$&$AD_n$ & $T_{n,4}$ & $R_n$ & $W_n$ & $S_{n,1}$ &  $S_{n,2}$& $S_{n,3}$& $S_{n,4}$ &$S_{n,5}$& $S_{n,6}$ & $S_{n,7}$ \\
		\hline
  $PP(1,2)$&7 &6 &6 &6 &4 &3 &5 &6 &6 &3 &4 &5
  \\ 
  $PP(1,1.5)$&5& 4 &3 &4 &4 &4  &4 &3 &4 &4 &4 &4  \\ 
  $PP(1,1)$&6&5&4&5&5&4&4&5&5&3&5&5 \\ \hline
	$BB(7,1)$ &95& 94 &97&84& \textbf{100} &\textbf{86}  &83 &72 &73 &{80} &73 &67  \\ 	
	$DU(15)$&94 &\textbf{98} &87 &60 &\textbf{98} &\textbf{70} &43 &29 &42 &{68} &17 &41\\ 
	$MPDU(10,0.25)$ &75&83 &31 &9 &\textbf{88} &87  &81 &48 &31 &\textbf{88} &64 &41 \\ 
  $MPDU(10,0.5)$&63 &\textbf{71} &2 & 14 &45 &\textbf{83}  &81 &72 &73 &80 &73 &67
\\ 
$MPBDU(1,3,0.75,3,0.25)$&66&\textbf{67} &21 &15 &54 &\textbf{94} &\textbf{94} &90 &69 &92 &\textbf{94} &87  \\ 
$MPBDU(2,3,0.75,3,0.25)$&82& \textbf{84} &69 &61 &82 &95  &96 &96 &81 &84 &\textbf{96} &94 \\ 
 $PB(1,3,0.75)$&39 &\textbf{39}&19 &29 &27 &33 &36 &51 &\textbf{58} &14 &38 &53 \\ 
  $PB(2,3,0.75)$&9&{16} &13 &\textbf{21} &7 &8 &10 &18 &\textbf{32} &6 &11 &19 \\ 
$NB(2,0.5)$&\textbf{16}& \textbf{16} &15 &\textbf{16} &15 &6 &6 &9 &\textbf{15} &4 &5 &10 \\ 
$NB(3,0.25)$&24& 31 &\textbf{34} &32 &33 &6 &10 &22 &\textbf{25} &6 &16 &24  \\ 	\hline
\end{tabular}
\end{table}

\begin{table}
	\centering
	\footnotesize
	\caption{{Testing Katz distribution: empirical powers (percentage or rejection rounded to the nearest integer $n=50,\alpha=0.05$).}}
	\label{Table03}
	\begin{tabular}{|c|ccccc|ccccccc|}
		\hline
		Alternative &$C_n$&$AD_n$ & $T_{n,4}$ & $R_n$ & $W_n$ & $S_{n,1}$  &$S_{n,2}$& $S_{n,3}$& $S_{n,4}$ &$S_{n,5}$& $S_{n,6}$ & $S_{n,7}$ \\
		\hline
   $Katz(2,0.5)$&5&4 & 5 & 4 & 4 & 3  &4 & 3 & 4 & 3 & 4 &
3 \\ 
   $P(1)$&6& 6 &4 &3 &4 &3 & 3 & 4 &5 & 2 & 3 & 5\\
$NB(4,0.25)$&4&4 &5 &4 &4 & 2 &4 &4 & 5 &4 &5 & 5\\\hline
  $PP(1,2)$&46&44 &42 &\textbf{48} &33 &6 &10 &16 &\textbf{34} &3 &10 &17 \\ 

 $PP(1,1)$&\textbf{14}&13 &11 &13 & 12 & 8 &8 &10 &\textbf{12} &2 & 8&10\\ 
$BB(2,2)$&23 &\textbf{24} &4 &9&22 & 5 &5 & 9 &\textbf{14} &3 &6 & 11 \\ 
		 $BB(6,2)$&27 &\textbf{36} &35 &32 &33 &5 &5 &4 &6&\textbf{7}  &5 &4 \\    
  $DU(2)$ & 60&\textbf{ 62} & 18 &35 &60 &22 &22 &35 &\textbf{46} &14 &24 &41  \\ 
$MKDU(8,0.5,2,0.5)$ &\textbf{100}&\textbf{100} &94 &80 &\textbf{100} &20  &88 &\textbf{100} &\textbf{100} &56 &98 & \textbf{100}\\ 
$MKDU(4,0.5,1,0.25)$&\textbf{63} &61 &5 &8 &42 &11  &50 &94 &\textbf{97} &17 &46 &96\\ 
$MKDU(4,0.5,2,0.25)$&\textbf{52} & 50 &12 &20 &23 &20  &55 &84 &80 &16 &60 &\textbf{86}\\ 
$MKP(4,0.5,1,0.25)$&\textbf{23} & 22 &5 &9 &10 &4 &9 &17 &\textbf{30} &9 &9 &19\\ 
$MNBP(4,0.25,1,0.25)$&\textbf{36}& 33 &7 &14 &20 &2 &8& 35 &\textbf{56} &9 &7 &41\\
  $MaxKDU(2,0.5,2)$&15&15 &\textbf{16} &11 &13 &6 &6 &13 &\textbf{17} &3 &6 &15\\ 
  $MaxKDU(2,0.5,8)$&21&\textbf{23} &10 &8 &20 &\textbf{10} &7 &3 &5 &\textbf{10} & 5 & 4 \\ 
 $PB(1,3,0.75)$&\textbf{87}& 86 &66 &80 &61 &33 &38 &65 &\textbf{91} &10 &36 &69 \\ 
$PB(2,3,0.75)$&17&30 &39 &\textbf{53} &11 &7 &10 &14 &\textbf{51} &6 &10 &16   \\ 
\hline
\end{tabular}
\end{table}

\begin{table}
	\centering
	\footnotesize
	\caption{{Testing Katz distribution: empirical powers (percentage or rejection rounded to the nearest integer $n=100,\alpha=0.05$).}}
	\label{Table04}
	\begin{tabular}{|c|ccccc|ccccccc|}
		\hline
		Alternative &$C_n$&$AD_n$ & $T_{n,4}$ & $R_n$ & $W_n$ & $S_{n,1}$  &$S_{n,2}$& $S_{n,3}$& $S_{n,4}$ &$S_{n,5}$& $S_{n,6}$ & $S_{n,7}$ \\
		\hline
   $Katz(2,0.5)$&5& 4 &4 &3 &4 &4  &4 &4 &4 &5 &4 &4
 \\ 
   $P(1)$&6&5 &5 &5 &5 &5 &5 &5 &6 &4 &5 &5\\
$NB(4,0.25)$&5&5 &6 &4 &5 &4  &4 &5 &5 &3 &4 &5 \\\hline
  $PP(1,2)$&\textbf{74}&73 &63 &71 &53 &15  &18 &39 &\textbf{70} &6 &16 &43 \\ 
  $PP(1,1)$&\textbf{22} &\textbf{22} &19 &21 &21 &10 &10 &15 &\textbf{19} &3 &10 &15\\ 
$BB(2,2)$ &43& \textbf{46} &19 &27 &44 &16 &18 &28 &\textbf{37} &12 &21 &34 \\ 
  $BB(6,2)$&56 & \textbf{72}&70&59 & 70&13 & 12 &7 &9 & \textbf{19} &11 & 6 \\ 
$DU(2)$&90 &\textbf{ 92} &71 &81 &90 &60  &69 &81 &\textbf{87} &54 &73 &85  \\ 
 $MKDU(8,0.5,2,0.5)$&\textbf{100} & \textbf{100} &\textbf{100} &\textbf{100} &97 &\textbf{100} &\textbf{100} &\textbf{100} &\textbf{100} &96 &\textbf{100} &\textbf{100}\\ 
$MKDU(4,0.5,1,0.25)$ &93& \textbf{94} &6 &8 &74 &55&97 &\textbf{100} &\textbf{100} &37 &98 &\textbf{100}\\ 
  $MKDU(4,0.5,2,0.25)$&\textbf{87} & 86 &18 &32 &39 &66 &97 &\textbf{100} &\textbf{100} &41 &99 &\textbf{100} \\ 
  $MKP(4,0.5,1,0.25)$&\textbf{39} &  \textbf{39} &4 &10 &21 &13 &31 &45 &\textbf{54} &19 &27 &46\\ 
 $MKP(4,0.5,2,0.25)$&\textbf{64} & \textbf{64} &45 &55 &35 &16 &21 &39 &\textbf{53} &11 &23 &43 \\ 
$MNBP(4,0.25,1,0.25)$&\textbf{61}&60 &9 &19 &43 &7 &39 &79 &\textbf{84} &18 &47 &80
\\
  $MaxKDU(2,0.5,2)$&27 &31 &28 &\textbf{32} &21 &12 &11 &25 &\textbf{34} &6 &11 &30 \\ 
  $MaxKDU(2,0.5,8)$&37&\textbf{46} &15 &11 &39 &21 &17 &5 &6 &\textbf{25} &10 &3 \\ 
 $PB(1,3,0.75)$&\textbf{100}& 99 &89 &97 &93 &77 &76 &99 &\textbf{100} &26 &72 &99 \\ 
$PB(2,3,0.75)$&33& 63 &66 &\textbf{82} &22 &13 &15 &35 &\textbf{86} &8 &16 &46  \\ \hline
\end{tabular}
\end{table}

Based on the results given in Tables \ref{Table01}-\ref{Table04} we might conclude that the impact of the weight function is significant. {This means that  proper selection of the weight function} should be made.  We highlight that 
 in several cases  the enhancement in the performance of the proposed novel tests, when weight terms $w_k$ are determined by the pmf of a Negative Binomial distribution, surpasses the scenario where no weight terms are used (e.g., the $S_{n,1}$ test with $w_k=1$). This improvement is particularly notable in cases such as $PB(1,3,0.75)$ and {$MKDU(4,0.5,1,0.25)$,} when dealing with the goodness-of-fit problem for Poisson-Poisson and Katz distributions, respectively.

In addition, Tables \ref{Table01}-\ref{Table04} reveal that in many cases, some of the novel tests either outperform or match the performance of existing tests. Specifically, for the Katz distribution, $S_{n,4}$ excels or matches existing tests when sampling from most tried alternative distributions; in other cases $S_{n,5}$ has better performance {in comparison with the others novel tests}. Similarly, for the Poisson-Poisson distribution, $S_{n,4}$ and $S_{n,5}$ perform well when certain distributions are sampled. These findings suggest that these novel tests could be valuable additions to the existing goodness-of-fit tests for Katz and Poisson-Poisson distributions. 
However, 
for the BB and DU alternatives, existing tests are more {powerful than the new ones.}

 In practice, the alternative distribution is unknown, and hence the user {does not know} which of  $S_{n,4}$, $S_{n,5}$
should be applied. 
With the aim  of trying to provide a rule for choosing between them, we did the following experiment:
we generated 10,000 samples,  with  $n=1000$, encompassing the set of alternatives employed in our simulation study in each gof problem. The  average values of $\big|\hat{d}\big(k,\hat{\lambda}, \hat{\theta}\big)\big|$ are reported in Tables \ref{Table05} and \ref{Table06}, which also display information about the maximum  average value {of $\big|\hat{d}\big(k,\hat{\lambda}, \hat{\theta}\big)\big|$} and the corresponding point where such maximum is attached. For clarity, we have divided both tables into two sections based on whether $S_{n,4}$ or $S_{n,5}$ exhibited higher power in the previous simulation experiment (Tables \ref{Table01}-\ref{Table04}). 
Recall also that when
 using as weights  the pmf of a  $NB(2,0.75)$ law, more importance is given to the terms $k=0\text{ }(56,25\%)$,   $k=1 \text{ }(28,125\%)$ and $k=2\text{ }(10,547\%)$; while using as weight term the pmf of the law  $NB(4,0.75)$, terms with greater $k$ are also significant, i.e.  $k=0 \text{ }(31,641\%)$,   $k=1\text{ }(31,641\%)$, $k=2\text{ }(19,775\%)$ and $k=3 (9,888\%)$.

Two common features among the alternatives where the gof test based on $S_{n,4}$ does not perform well  for both the Poisson-Poisson and Katz distributions (as seen in the second part of Tables \ref{Table05} and \ref{Table06}) is that the average absolute value of $\hat{d}(0,\hat{\lambda},\hat{\theta})$ is significantly smaller when compared to the maximum value and the other values, and the maximum value is reached for $k> 2$. Among the weight functions used in this simulation study, the pmf of the law $NB(4,0.25)$
 assigns greater weight to larger terms, making $S_{n,5}$ a 
 powerful test statistic. The opposite is observed in the cases where the gof test based on $S_{n,4}$ shows larger power. 

\begin{table}
	\centering
	\footnotesize
	\caption{{GOF for {Poisson-Poisson} distribution: average values of $|d_k|$ for alternatives, $n=1000$.}}
	\label{Table05}
	\begin{tabular}{|c|cccccccccc|}
		\hline
        &\multicolumn{9}{c}{$k$}&\\
		Alternative & 0&1 & 2 & 3 & 4 & 5 & 6& 7 &8 & $\max|d_k|, k$ \\
		\hline
 $PB(1,3,0.75)$&0.136& 0.016& 0.209& 0.120& 0.054& 0.042& 0.043& 0.007& 0.005&0.209,2
\\ 
  $PB(2,3,0.75)$&0.101& 0.002& 0.114& 0.112& 0.085& 0.056& 0.039& 0.046& 0.002&0.114,2\\ 
$NB(2,0.5)$& 0.065& 0.007& 0.038& 0.039& 0.023& 0.006& 0.004& 0.007& 0.007&0.065,0 \\ 
$NB(3,0.25)$&0.028& 0.069& 0.095& 0.093& 0.066& 0.026& 0.014& 0.047& 0.068&0.095,2\\   
\hline
$BB(7,1)$ & 0.140& 0.148& 0.056& 0.063& 0.187& 0.313& 0.437& 0.438& 0.040&0.438,7\\
  $DU(15)$&  0.012& 0.087& 0.149& 0.164& 0.139& 0.090& 0.030& 0.031& 0.095&0.469,16\\
  $MPDU(10,0.25)$ & 0.046& 0.081& 0.191& 0.175& 0.071& 0.056&0.164 & 0.253& 0.332&0.410,9 \\ 
$MPDU(10,0.5)$ &0.150&0.026& 0.174& 0.232& 0.156& 0.026& 0.093& 0.187& 0.257& 0.311,9 \\ 
$MPBDU(1,3,0.75,3,0.25)$ & 0.075& 0.157& 0.283 &0.369& 0.227& 0.095& 0.021& 0.042& 0.040 &0.369, 3\\
$MPBDU(2,3,0.75,3,0.25)$ & 0.075& 0.157& 0.283 &0.369& 0.227& 0.095& 0.021& 0.042& 0.040 &0.369, 3\\\hline
\end{tabular}
\end{table}

\begin{table}
	\centering
	\footnotesize
	\caption{{GOF for Katz distribution: average values of $|d_k|$ for alternatives, $n=1000$.}}
	\label{Table06}
	\begin{tabular}{|c|cccccccccc|}
		\hline
        &\multicolumn{9}{c}{$k$}&\\
		Alternative & 0&1 & 2 & 3 & 4 & 5 & 6& 7 &8 & $\max|d_k|, k$ \\
		\hline
  $PP(1,2)$&0.167 &0.004& 0.063& 0.067& 0.050& 0.032& 0.019& 0.006& 0.002&0.167,0 \\  
  $PP(1,1)$&0.071& 0.037& 0.037& 0.017& 0.003& 0.004& 0.005& 0.005& 0.004&0.071,0\\ 
		$BB(2,2)$ &0.099& 0.265& 0.276& 0.184& 0.124& 0.084& 0.057& 0.039& 0.027&0.276,2  \\ 		
  $DU(2)$ & 0.166& 0.415& 0.375& 0.189& 0.096& 0.049& 0.026& 0.013& 0.007&0.415,1\\ 
$MKDU(8,0.5,2,0.5)$ &0.001& 0.021& 0.439& 0.383& 0.328& 0.271& 0.208& 0.140& 0.067&0.439,2\\ 
  $MKDU(4,0.5,1,0.25)$ & 0.235& 0.249& 0.183& 0.115& 0.050& 0.004& 0.045& 0.070& 0.081& 0.249,1 \\ 
  $MKDU(4,0.5,2,0.25)$ &  0.115& 0.266& 0.300& 0.202& 0.116& 0.045& 0.007& 0.041& 0.059&0.300,2\\ 
  $MKP(4,0.5,1,0.25)$ &0.129& 0.021& 0.116& 0.146& 0.104& 0.044& 0.006& 0.040& 0.057& 0.146,3\\ 
  $MKP(4,0.5,2,0.25)$ & 0.104& 0.143 &0.037& 0.093& 0.148& 0.129& 0.076& 0.026& 0.011& 0.148,4\\ 
$MNBP(4,0.25,1,0.25)$& 0.169& 0.077& 0.082& 0.151& 0.147& 0.115& 0.080& 0.046& 0.018& 0.169,0\\
  $MaxKDU(2,0.5,2)$& 0.094& 0.266& 0.255& 0.115& 0.053& 0.020& 0.005& 0.013& 0.015& 0.266,1\\ 
 $PB(1,3,0.75)$&0.284 &0.074& 0.245& 0.111& 0.109& 0.067& 0.033& 0.013& 0.009& 0.284,0 \\ 
$PB(2,3,0.75)$& 0.212& 0.026& 0.099& 0.137& 0.143& 0.083& 0.011& 0.086& 0.021&0.212,0\\ 
		\hline
$BB(6,2)$ & 0.107 & 0.072 & 0.036& 0.141& 0.180& 0.071& 0.250& 0.000& 0.000&0.250,6 \\ 
$MaxKDU(2,0.5,8)$& 0.001& 0.027& 0.054& 0.044& 0.011& 0.097& 0.201& 0.312& 0.528& 0.528,9 \\ \hline
\end{tabular}
\end{table}


\section{Real data illustrations}\label{application}
This section presents the application
 of the tests used in the simulation study of the previous section to six real data sets for testing gof to the Katz distribution. 

The first data set represents the number of claims of automobile liability policies (\cite{Klugman}, p. 244); the second data set corresponds to the number of chromatid aberrations in 24 hours (\cite{Catcheside1, Catcheside2}); the third data set  concerns  the number of lost articles found in the Telephone
and Telegraph Bldg., New York City (see \cite{Consul} and references therein); the fourth represents the number of hemocytometer yeast cell on European red mites on apple leaves (\cite{Rashid});  the fifth 
data set represents  the number of absences of workers in a particular division of a large steel corporation in an observational period of six months (\cite{Sichel});
 while the last one 
represents the number of borers per corn plant dissected in  corn plants growing in an area located in Northwest Iowa {(\cite{McGuire})}. 
Table \ref{pvalues} \textcolor{black}{lists  the maximum likelihood (ML) estimates of the parameters of the Katz distribution and the obtained bootstrap $p$-values (with $B=5000$) of the gof tests}.

All tests concur that  
the null hypothesis cannot be rejected, at the 5\% significance level, for the first four data sets.
However, a different scenario unfolds when assessing the gof of the Katz distribution for the fifth data set. Here, the unanimous agreement among the {five} existing tests ($CV_n$, $AD_n$, $T_{n,4}$, $R_n$, and $W_n$) that the Katz distribution is unsuitable at the {5\% significance level,  is solely supported  by $S_{n,4}$.}
As discussed in the previous section, the weight terms play a pivotal role in the performance of the novel tests introduced here. Recommendations for selecting these weight terms were based on the absolute values of $\hat{d}(k,\hat{\lambda},\hat{\theta})$, that for this data set are:  
0.0612, 0.0295, 0.0090, 0.0647, 0.0202, 0.0445, 0.0211, 0.0119, 0.0067, corresponding to $k=0,\ldots, 8$, respectively, which leads us to choose the test based on  $S_{n,4}$.
Recall that
$S_{n,4}$ was recommended for all alternative cases, expect 
when the absolute value of $\hat{d}(0,\hat{\lambda},\hat{\theta})$ is notably smaller in comparison to the maximum value, with the maximum occurring for $k> 2$. 

As for the sixth data set, with the exception of the {test based on $T_{n,4}$ and $W_n$}, all other tests agree that at significance level $5\%$ the Katz distribution is not adequate for fitting the borers data.

\begin{table}[!htp]
\small
\begin{center}
\caption{
\textcolor{black}{ML estimates for the parameters of the Katz distribution and bootstrap $p$-values of the gof tests.} 
}\label{pvalues}
\resizebox{\textwidth}{!}{
\begin{tabular}{|c|c|ccccc|ccccccc|}\hline
Data set &$(\hat{\lambda},\hat{\theta})$&$C_n$ &$AD_n$ & $T_{n,4}$ & $R_n$ & $W_n$ & $S_{n,1}$&$S_{n,2}$& $S_{n,3}$& $S_{n,4}$ &$S_{n,5}$& $S_{n,6}$ & $S_{n,7}$ \\\hline
I: Claims & (0.80,0.53)&0.32  & 0.33& 0.84 &0.67& 0.62& 0.26&0.36 &0.44 &0.27& 0.15 &0.49& 0.40\\\hline
II: Chromatid& (0.27,0.51) & 0.08& 0.08 &0.14& 0.10& 0.09& 0.18 &0.26 &0.23& 0.14& 0.16& 0.32& 0.22\\\hline
III: Articles&(0.85,0.18) &0.60& 0.69 &0.54 &0.49& 0.63 &0.86&  0.87& 0.84 &0.71 &0.84 &0.88 &0.81\\\hline
IV: Cell&(0.63,0.45)&0.17& 0.94 &0.12 &0.11 &0.16& 0.62  &0.61& 0.43& 0.32& 0.79& 0.58 &0.39\\\hline
V: Absence& (0.29,0.56)&0.03 & 0.03 &0.03 &0.03& 0.03& 0.07&  0.07 &0.07 &0.04&  0.13& 0.08& 0.08\\\hline
VI: Borers& (2.05,0.21) &0.02& 0.009& 0.05 &0.02 &0.07& 0.004&  0.007&0.003 &0.002 & 0.01& 0.008 &0.002\\\hline
\end{tabular}
}
\end{center}
\end{table}

\section{Conclusions}\label{conclusions}
This paper proposes a unified approach for testing gof to any distribution which can be viewed as a generalization of the Poisson distribution, in the sense that its pgf is the only one that satisfies a certain differential equation. The new gof test statistics are a function of the coefficients of the polynomial
of the resulting equation when  the pgf is replaced with the epgf in the aforementioned differential equation, multiplied by positive constants which serve as weight terms. The null distribution can be consistently approximated by a parametric bootstrap. The finite sample performance of the new tests were numerically assessed through an extensive simulation experiment, which includes some existing tests with the aim of comparing their powers.
The numerical results reveal that {no test
provides the highest power against all alternatives considered: for some alternatives,
the new test exhibits the highest power, but for others the competing tests yield greater power. 
}  It is worth emphasizing
that the performance of the new test is improved with the introduction of the weight terms. Choosing them carefully results in  better outputs in terms of power.


{
\section*{Acknowledgment}
We express our gratitude to the anonymous referees for their valuable comments and remarks that improved the paper.
}
\section*{Disclosure statement}
The authors report there are no competing interests to declare.

\section*{Funding information}
A. Batsidis acknowledges support of this work by the project Establishment of capacity building infrastructures in Biomedical Research (BIOMED-20) (MIS 5047236) which is implemented under the Action  Reinforcement of the Research and Innovation Infrastructure, funded by the Operational Programme Competitiveness, Entrepreneurship and Innovation (NSRF 2014-2020) and co-financed by Greece and the European Union (European Regional Development Fund).
The work of B. Milo\v sevi\'c is  supported by the Ministry of Science, Technological Development and Innovations of the Republic of Serbia (the contract 451-03-66/2024-03/200104),  and  by the COST action
CA21163 - Text, functional and other high-dimensional data in econometrics: New models, methods, applications (HiTEc).
M.D. Jim\'enez-Gamero  is supported by grants PID2020-118101GB-I00 and  PID2023-148811NB-I00, funded by MICIU/AEI/10.13039/501100011033 and ERDF/EU.

\appendix
\section{APPENDIX: Proofs}

Recall the definition of $\phi(x;k,\lambda,\theta)$ in \eqref{phik} and let
  $\phi(x;\lambda,\theta)=(\phi(x; 0,\lambda,\theta),\phi(x; 1,\lambda,\theta), \ldots)$.

\begin{lemma}\label{lemma1}
Let  $X_1,\ldots, X_n$ be iid from a  random variable $X \in \mathbb{N}_{0}$, with $E(X^2)<\infty$.

Assume that $w$ satisfies \eqref{w}, and that Assumption \ref{Q} (i) holds. Then
$
E(\| \phi(X;\lambda,\theta)\|_w^2)  <\infty.
$
\end{lemma}
\begin{proof}
By definition,
\begin{eqnarray*}
\| \phi(X;\lambda,\theta)\|_w^2   & =  & \sum_{k \geq 0}\phi(X; k,\lambda,\theta)^2 w_k \\
                          & =  & \sum_{k \geq 0}w_k(k+1)^2I(X=k+1)
                          + \lambda^2\sum_{k \geq 0}w_k
                                  \sum_{u= 0}^{k}q_{k-u}^2(\theta)I(X=u),
\end{eqnarray*}
and thus
\begin{eqnarray*}
E(\| \phi(X;\lambda,\theta)\|_w^2)  
 & \leq & M  \sum_{k \geq 0}(k+1)^2p_{k+1}+M\lambda^2 \sum_{k \geq 0} \sum_{u= 0}^{k}q_{k-u}^2(\theta)p_{u}.
\end{eqnarray*}
Notice that Assumption \ref{Q} (i)  implies that $\sum_{k \geq 0} q_{k}^2(\theta)<\infty$. Since $\sum_{k \geq 0}(k+1)^2p_{k+1} = E(X^2)<\infty,$ and $\sum_{k \geq 0} \sum_{u= 0}^{k}q_{k-u}^2(\theta)p_{u}=\sum_{k \geq 0} q_{k}^2(\theta)<\infty,$ the result follows.
\end{proof}

\begin{lemma}\label{lemma2}
Let  $X_1,\ldots, X_n$ be iid from a  random variable $X \in \mathbb{N}_{0}$.
Assume that $w$ satisfies \eqref{w}, and that Assumption \ref{Q} (i) and (ii) hold.
Then
\begin{itemize} \itemsep=0pt
\item[(i)] $\displaystyle \left \|\frac{\partial}{\partial \lambda}\widehat{d}(\cdot;\lambda,\theta)\right\|_w^2 <\infty$,
$\displaystyle \left \|\frac{\partial}{\partial \theta_j}\widehat{d}(\cdot;\lambda,\theta) \right\|_w^2 <\infty$, $1 \leq j\leq p$, with probability 1.

\item[(ii)]$\displaystyle \left \|E\left\{\frac{\partial}{\partial \lambda}\widehat{d}(\cdot; \lambda,\theta)\right\} \right \|_w^2<\infty$,
$\displaystyle \left \|E\left\{\frac{\partial}{\partial \theta_j}\widehat{d}(\cdot;\lambda,\theta)\right\}\right\|_w^2<\infty$, $1 \leq j \leq p$.
\item[(iii)] $\displaystyle \left \|\frac{\partial}{\partial \lambda}\widehat{d}(\cdot;\lambda,\theta)-E\left\{\frac{\partial}{\partial \lambda}\widehat{d}(\cdot; \lambda,\theta)\right\}\right\|_w \stackrel{a.s.}{\longrightarrow}0$,
    $\displaystyle \left \|\frac{\partial}{\partial \theta_j}\widehat{d}(\cdot;\lambda,\theta)-E\left\{\frac{\partial}{\partial \theta_j}\widehat{d}(\cdot; \lambda,\theta)\right\} \right \|_w \stackrel{a.s.}{\longrightarrow}0$,  $1 \leq j \leq p$.
\item[(iv)] If Assumption \ref{Q} (iii)  holds,  $|\lambda_k -\lambda| \leq M_2$, for some positive constant $M_2$,  and $\theta_k\in \mathcal{N}(\theta)$, $\forall k \geq 0$, then
 $\displaystyle  \|\hat{D}_0\|_w^2<\infty$,  $\displaystyle  \|\hat{D}_j\|_w^2<\infty$, $1 \leq j \leq p$, \red{with probability 1}, where
\begin{eqnarray*}
 \hat{D}_0 & = & \left(  \frac{\partial}{\partial \lambda}\hat{d}(0;\lambda_0,\theta_0) -
  \frac{\partial}{\partial \lambda}\hat{d}(0;\lambda,\theta),   \frac{\partial}{\partial \lambda}\hat{d}(1;\lambda_1,\theta_1) -
  \frac{\partial}{\partial \lambda}{\hat{d}(1;\lambda,\theta)} , \ldots \right ),\\
  \hat{D}_j & = & \left(  \frac{\partial}{\partial \theta_j}\hat{d}(0;\lambda_0,\theta_0) -
  \frac{\partial}{\partial \theta_j}\hat{d}(0;\lambda,\theta),   \frac{\partial}{\partial \theta_j}\hat{d}(1;\lambda_1,\theta_1) - \frac{\partial}{\partial \theta_j}{\hat{d}(1;\lambda,\theta)}, \ldots \right ), \quad 1 \leq j \leq p.
\end{eqnarray*}
\item[(v)] {If Assumption \ref{Q} (iv) holds, for each $k \geq 0$,
  $\left\{\lambda_{kn} \right\}_{n\geq 1}$ is a stochastic sequence such that $\displaystyle \sup_{k \geq 0}|\lambda_{kn}-\lambda| \stackrel{P}{\longrightarrow} 0$,  for each $k \geq 0$,
 $\left\{\theta_{kn} \right\}_{n\geq 1}$ is a stochastic sequence  such that $ \displaystyle \sup_{k \geq 0}\|\theta_{kn} -\theta\|\stackrel{P}{\longrightarrow}  0$,  and $\lambda_k$, $\theta_k$ are  replaced with $\lambda_{kn}$, $\theta_{kn}$, respectively, in the expression of  $ \hat{D}_j$, $0 \leq j \leq 0$, then 
 $\displaystyle  \|\hat{D}_j\|_w^2 
\stackrel{P}{\longrightarrow} 0$, $0 \leq j \leq p$.}
\end{itemize}
\end{lemma}

\begin{proof}
(i) We have that
\begin{equation} \label{partial.lambda}
\frac{\partial}{\partial \lambda}\widehat{d}(k;\lambda,\theta)=-\sum_{u=0}^{k}q_u(\theta)\hat{p}_{k-u},
\end{equation}
and therefore
\begin{eqnarray*}
\left \|\frac{\partial}{\partial \lambda}\widehat{d}(\cdot;\lambda,\theta) \right \|_w^2
&    = &   \sum_{k\geq 0}w_k\sum_{u,v=0}^{k}q_u(\theta)q_v(\theta) \hat{p}_{k-u}\hat{p}_{k-v}\\
& \leq & M \sum_{k\geq 0}\sum_{u,v=0}^{k} \left|q_u(\theta)\right| \left|q_v(\theta) \right|  \hat{p}_{k-u}\hat{p}_{k-v}\\\
& = & M \sum_{u,v \geq 0} \left|q_u(\theta)\right| \left|q_v(\theta) \right|  \sum_{k \geq 0}\hat{p}_k\hat{p}_{k+|u-v|}.
\end{eqnarray*}
Taking into account  
\begin{equation} \label{tia}
\sum_{k \geq 0}\hat{p}_k\hat{p}_{k+|u-v|} \leq \sum_{k \geq 0}\hat{p}_k=1,
\end{equation}
we have that
\[
\left \|\frac{\partial}{\partial \lambda}\widehat{d}(\cdot;\lambda,\theta) \right \|_w^2 \leq
M \sum_{u,v \geq 0} \left|q_u(\theta)\right| \left|q_v(\theta) \right|  =M \bigg(\sum_{u \geq 0} \left|q_u(\theta)\right|   \bigg)^2< \infty,
\]
where the  finiteness comes from  Assumption \ref{Q}(i).

 We have that
\[
\frac{\partial}{\partial \theta_j}\widehat{d}(k;\lambda,\theta)=-
\red{\lambda}\sum_{u=0}^{k}\frac{\partial}{\partial \theta_j}q_u(\theta)\hat{p}(k-u).
\]
Proceeding as before, one gets that
\[
\left \|\frac{\partial}{\partial  \theta_j}\widehat{d}(\cdot;\lambda,\theta) \right \|_w^2 \leq
M  \left(\sum_{u \geq 0} \left| \frac{\partial}{\partial \theta_j} q_u(\theta) \right| \right)^2< \infty,
\]
where the finiteness comes from  Assumption \ref{Q}(ii).

(ii) The result follows from part (i) by replacing $\hat{p}_{k}$ with $p_{k}$, $\forall k\geq 0$.

(iii) From \eqref{prob} and \eqref{partial.lambda}, we can write
\[
\frac{\partial}{\partial \lambda}\widehat{d}(k;\lambda,\theta)=-\frac{1}{n}
\sum_{i=1}^{n}
\sum_{u=0}^{k}q_u(\theta)I(X_i=k-u),
\]
and thus,
$$
\frac{\partial}{\partial \lambda}\widehat{d}(\cdot;\lambda,\theta)=-\frac{1}{n}
\sum_{i=1}^{n} \iota(X_i),
$$
with
$$\iota(X)=(q_0(\theta)I(X=0), \sum_{u=0}^{1}q_u(\theta)I(X_i=1-u),\ldots, \sum_{u=0}^{k}q_u(\theta)I(X_i=k-u), \ldots).$$
Notice  that $\iota(X)=(0,\ldots,0,q_0(\theta), q_1(\theta), q_2(\theta), \ldots)$, where $q_0(\theta)$ is in the position $X$ (starting at 0). Hence,
$\|\iota(X)\|_w^2=\sum_{k \geq 0}w_{X+k}q_k(\theta)^2 \leq M \sum_{k \geq 0}q_k(\theta)^2<\infty$, which implies that $E\{\|\iota(X)\|_w^2 \}<\infty$. Now, by the SLLN in Banach spaces (see Theorem \red{2.4} in Bosq \cite{Bosq}) and the continuous mapping theorem, it follows that
$\displaystyle \Big \|\frac{\partial}{\partial \lambda}\widehat{d}(\cdot;\lambda,\theta)-E\Big\{\frac{\partial}{\partial \lambda}\widehat{d}(\cdot; \lambda,\theta)\Big\}\Big\|_w \stackrel{a.s.}{\longrightarrow}0$.

The proof for the derivatives with respect to $\theta_j$, $1 \leq j \leq p$, is similar, so we omit it.

{(iv) Proceeding as in part (i), one gets that
\[
\|D_0\|_w^2 \leq M\left(\sum_{u \geq 0} \left|q_u(\theta_k)-q_u(\theta) \right| \right)^2.
\]
By Assumption \ref{Q} (iii) it follows that $\|D_0\|_w^2 <\infty$, with probability one.  The proof for $\|D_j\|_w^2$, $1 \leq j \leq p$, is similar, and thus omitted.} 

{(v) Proceeding as in part (i), one gets that
\[
\|D_0\|_w^2 \leq M\left(\sum_{u \geq 0} \left|q_u(\theta_{kn})-q_u(\theta) \right| \right)^2.
\]
From Assumption \ref{Q} (iv), for any $\varepsilon>0$, there exists  $\delta>0$ and $n_0\in \mathbb{N}$ such that if $\displaystyle \sup_{k \geq 0} |\theta_{kn}-\theta|<\delta$ and $n>n_0$, then $\sum_{u \geq 0} \left|q_u(\theta_{kn})-q_u(\theta) \right|<\varepsilon$. Therefore, for large $n$,
\[P\left (\|D_0\|_w^2 > M\varepsilon \right) \leq  P \left(\sum_{u \geq 0} \left|q_u(\theta_{kn})-q_u(\theta) \right|>\varepsilon \right) \leq 
P\left( \sup_{k \geq 0} |\theta_{kn}-\theta|>\delta\right).\]
By assumption, the  right-most term of the above expression goes to 0. This proves that $\|D_0\|_w^2\stackrel{P}{\longrightarrow}0$.  The proof for $\|D_j\|_w^2$, $1 \leq j \leq p$, is similar, and thus omitted.}  
 \end{proof}

\begin{proof}[Proof of Theorem \ref{power1}]
By Taylor expansion, we get that, for each $k\in \mathbb{N}_{0}$,
\begin{eqnarray}
\hat{d}(k;\hat{\lambda},\hat{\theta})& = &\hat{d}(k;\lambda,\theta)  +
\frac{\partial}{\partial \lambda}\hat{d}(k;{\lambda}, \theta)(\hat{\lambda}-\lambda)+
\sum_{i=1}^p \frac{\partial}{\partial \theta_i}\hat{d}(k;\lambda,\theta)(\hat{\theta}_i-\theta_i) \nonumber\\
 & & +\left( \frac{\partial}{\partial \lambda}\hat{d}(k;\bar{\lambda}_k,\bar{\theta}_k) -
  \frac{\partial}{\partial \lambda}\hat{d}(k;\lambda,\theta) \right)  (\hat{\lambda}-\lambda  ) \nonumber\\
 & & +\sum_{i=1}^p \left( \frac{\partial}{\partial \theta_i}\hat{d}(k;\bar{\lambda}_k,\bar{\theta}_k) -
  \frac{\partial}{\partial \theta_i}\hat{d}(k;\lambda,\theta) \right)  (\hat{\theta}_i-\theta_i  ) \label{taylor}
\end{eqnarray}
with $(\bar{\lambda}_k,\bar{\theta}_k^\top)=a_k (\lambda,\theta^{\top})+(1-a_k)(\hat{\lambda},\hat{\theta}^{\top})$, for some $a_k \in (0,1)$.
Let $\hat{D}_0$,
  $\hat{D}_i$, $1 \leq i \leq p$, be as defined in Lemma  \ref{lemma2} (iv) with $\lambda_k=\bar{\lambda}_k$, $\forall k \geq 0$.
Therefore,
\begin{eqnarray}
\red{\left| \|\hat{d}( \cdot ;\hat{\lambda},\hat{\theta})\|_w - 
 \|\hat{d}(  \cdot;\lambda,\theta) \|_w \right|} & \leq & 
 \left \| \frac{\partial}{\partial \lambda} \hat{d}(\cdot;\lambda,\theta) \right \|_w \left |\hat{\lambda}-\lambda \right |+
\sum_{i=1}^p\left \|\frac{\partial}{\partial \theta_i}\widehat{d}(\cdot;\lambda,\theta) \right \|_w \left |\hat{\theta}_i-{\theta}_i \right | \nonumber\\
 & & +   \|\hat{D}_0\|_w   \left |\hat{\lambda}-\lambda \right |+
 \sum_{i=1}^p   \|\hat{D}_i\|_w    \left |\hat{\theta}_i-{\theta}_i \right |. \label{aux1}
  \end{eqnarray}
Since
\begin{equation} \label{d.average}
\hat{d}(k;{\lambda},{\theta})=\frac{1}{n}\sum_{i=1}^{n}\phi(X_{i};k,{\lambda},{\theta}),
\end{equation}
taking into account that, from Lemma \ref{lemma1},  $E(\| \phi(X;\lambda,\theta)\|_w^2) <\infty$,  { and that $E\{\phi(X;\lambda,\theta)\}=d(\cdot;\lambda,\theta)$,} then by applying the SLLN in Banach spaces (see Theorem \red{2.4}  in Bosq \cite{Bosq}) and the continuous mapping theorem, it follows that
\begin{equation} \label{aux2}
\| \hat{d}( \cdot;\lambda,\theta)\|_w^2 \stackrel{a.s.}{\longrightarrow} { \| E\{\phi(X;\lambda,\theta)\}\|_w^2=\| {d}( \cdot;\lambda,\theta)\|_w^2=\eta}<\infty.
\end{equation}
{From  Lemma \ref{lemma2} (i), the quantities 
$ \left \| \frac{\partial}{\partial \lambda} \hat{d}(\cdot;\lambda,\theta) \right \|_w${,} $\left \|\frac{\partial}{\partial \theta_1}\widehat{d}(\cdot;\lambda,\theta) \right \|_w, \ldots, \left \|\frac{\partial}{\partial \theta_p}\widehat{d}(\cdot;\lambda,\theta)\right \|_w$ are all of them bounded with probability 1.
Since  $\hat{\lambda} \stackrel{a.s.(P)}{\longrightarrow} \lambda$,  it readily follows that
$|\hat \lambda-\lambda| \leq M_2$,  for large enough $n$ with probability 1 (with high probability), for some positive constant $M_2$. This fact implies that
$|\bar{\lambda}_k-\lambda| \leq M_2$,
$\forall k \geq 0$, for large enough $n$ with probability 1 (with high probability).
Analogously, as $\hat{\theta} \stackrel{a.s.(P)}{\longrightarrow} \theta$, it readily follows that $\bar{\theta}_{(k)} \in {\cal N}(\theta)$, $\forall k \geq 0$, for large enough $n$ with probability 1 (with high probability). As a consequence, from 
 Lemma \ref{lemma2} (iv), the quantities 
 $ \|\hat{D}_0\|_w$, 
 $\|\hat{D}_1\|_w, \ldots,  \|\hat{D}_p\|_w$ are all of them bounded with probability 1 (with high probability)}. Finally, the result follows from \eqref{aux1}, \eqref{aux2}, \red{the previous observation and that  {$\hat{\lambda} \stackrel{a.s.(P)}{\longrightarrow} \lambda$}, $\hat{\theta} \stackrel{a.s.(P)}{\longrightarrow} \theta$.}
\end{proof}

\begin{proof}[Proof of Theorem \ref{asymptoticnulldistribution}]
 Let $d(k;\lambda,\theta)$ be as defined in the statement of Theorem \ref{power1}. Let us first notice that
\begin{eqnarray}
& & \frac{\partial}{\partial \lambda}d(k;\lambda,\theta)  =  E_{\lambda, \theta}\left\{\frac{\partial}{\partial \lambda}\widehat{d}(\cdot; \lambda,\theta)\right\}  
=E_{\lambda, \theta} \left\{   \frac{\partial}{\partial \lambda}\phi(X; k,\lambda,\theta) \right\}=\mu_0(k;\lambda,\theta),  \label{caca0} \\ & &
 \frac{\partial}{\partial \theta_i}d(k;\lambda,\theta) =  E_{\lambda, \theta}\left\{\frac{\partial}{\partial \theta_j}\widehat{d}(\cdot;\lambda,\theta)\right\}= 
 E_{\lambda, \theta} \left \{  \frac{\partial}{\partial \theta_j} \phi(X; k,\lambda,\theta) \right\}=\mu_j(k;\lambda,\theta), \quad 1 \leq j \leq p.  \label{cacai}
\end{eqnarray}
 From expansion \eqref{taylor},  \eqref{d.average},  and Assumption \ref{estim}, we can write
\begin{eqnarray}
\sqrt{n}\hat{d}(k;\hat{\lambda},\hat{\theta}) & = & \frac{1}{\sqrt{n}} \sum_{i=1}^n Y(X_i; k, \lambda,\theta)
\label{auxn.1}\\
 & & +\left\{ \frac{\partial}{\partial \lambda}\hat{d}(k;\lambda,\theta) -\frac{\partial}{\partial \lambda}d(k;\lambda,\theta)\right\} \frac{1}{\sqrt{n}} \sum_{i=1}^n \psi_0(X_i; \lambda, \theta) \label{auxn.2}\\
 & & + \sum_{j=1}^p \left\{ \frac{\partial}{\partial \theta_j}\hat{d}(k;\lambda,\theta) -\frac{\partial}{\partial \theta_j}d(k;\lambda,\theta)\right\} \frac{1}{\sqrt{n}} \sum_{i=1}^n \psi_0(X_i; \lambda, \theta)\label{auxn.3}\\
 & & + \frac{\partial}{\partial \lambda}\hat{d}(k;\lambda,\theta) r_0+\sum_{j=1}^p \frac{\partial}{\partial \theta_j}\hat{d}(k;\lambda,\theta) r_{j} \label{auxn.4}\\
 & & + \left( \frac{\partial}{\partial \lambda}\hat{d}(k;{\bar{\lambda}_{k},\bar{\theta}_{k}}) -
  \frac{\partial}{\partial \lambda}\hat{d}(k;\lambda,\theta) \right) \sqrt{n}(\hat{\lambda}-\lambda) \label{auxn.5}\\
 & &+\sum_{j=1}^p \left( \frac{\partial}{\partial \theta_i}\hat{d}(k;{\bar{\lambda}_{k},\bar{\theta}_{k}}) -
  \frac{\partial}{\partial \theta_i}\hat{d}(k;\lambda,\theta) \right)\sqrt{n}(\hat{\theta}_j-\theta_j), \label{auxn.6}
\end{eqnarray}
{where $Y(X; k, \lambda, \theta)$ is as defined in \eqref{Y.expression}.}
 To show the result we prove that the norm of the sequence whose elements are in each of the equations \eqref{auxn.2}--\eqref{auxn.6} converge in probability to 0, and that the   sequence of  elements in equation \eqref{auxn.1} converges in law to a Gaussian element taking values in $l^2_w$ with mean zero and covariance as defined in the statement of the result.

  From Assumption \ref{estim},
  $(1/\sqrt{n}) \sum_{i=1}^n \psi_0(X_i; \lambda, \theta)$ converges in law to a normal law with mean zero and finite variance, which implies that it is bounded in probability; by Lemma \ref{lemma2} (iii), the norm of the sequence whose elements are in curly brackets in \eqref{auxn.2}, goes to 0 a.s. Therefore, the norm of the sequence whose elements are in \eqref{auxn.2} goes to 0 in probability.

The same reasoning applies to show that the norm of the sequence whose elements are in \eqref{auxn.3} goes to 0 in probability.

\red{From Lemma \ref{lemma2} (i),  $ \left \|\frac{\partial}{\partial \lambda}\widehat{d}(\cdot;\lambda,\theta)\right\|_w^2 <\infty$,
$ \left \|\frac{\partial}{\partial \theta_j}\widehat{d}(\cdot;\lambda,\theta) \right\|_w^2 <\infty$, $1 \leq j\leq p$, with probability 1,  and from Assumption \ref{estim},  $r_{j} \stackrel{P}{\longrightarrow} 0,$  $0 \leq j \leq p$. Hence, it readily follows that  the norm of the sequence whose elements are in \eqref{auxn.4} goes to 0 in probability}. 

{
From Assumption \ref{estim}, by applying the central limit theorem to $(1/\sqrt{n})\sum_{i=1}^n \psi_j(X_i; \lambda, \theta)$, $0 \leq j \leq p$, and Slutsky's theorem, if follows  that all of  $\sqrt{n}(\hat{\lambda}-\lambda)$, $\sqrt{n}(\hat{\theta}_1-\theta_1), \ldots, \sqrt{n}(\hat{\theta}_p-\theta_p)$  converge in law to normal laws, and therefore, all those quantities are bounded in probability. Moreover, it also entails that
 $\hat{\lambda} \stackrel{P}{\longrightarrow} \lambda$ and $\hat{\theta} \stackrel{P}{\longrightarrow} \theta$, which implies that $\displaystyle \sup_{k \geq 0} |\bar{\lambda}_k-\lambda| \stackrel{P}{\longrightarrow} 0$ and  $\displaystyle \sup_{k \geq 0} \|\bar{\theta}_k-\theta\| \stackrel{P}{\longrightarrow} 0$.
From Lemma \ref{lemma2} (v), it follows that the  norm of the sequence whose elements are $\frac{\partial}{\partial \lambda}\hat{d}(k;{\bar{\lambda}_{k},\bar{\theta}_{k}}) -
  \frac{\partial}{\partial \lambda}\hat{d}(k;\lambda,\theta) $ goes to 0 in probability, and the same happens for the
 norm of the sequence whose elements are 
 $  \frac{\partial}{\partial \theta_i}\hat{d}(k;{\bar{\lambda}_{k},\bar{\theta}_{k}}) -
  \frac{\partial}{\partial \theta_i}\hat{d}(k;\lambda,\theta) $, $1\leq i \leq p$.  As a consequence of the above facts,  the norm of the sequence whose elements are in   \eqref{auxn.5} and  \eqref{auxn.6}
goes to 0 in probability.} 

{From Assumption \ref{estim} and \eqref{relation}, it readily follows that  $E_{\lambda, \theta}\{Y(X; k,\lambda,\theta)\}=0$, $\forall k \geq 0.$}
Let
\begin{equation} \label{Y}
Y(X; \lambda, \theta)= ( Y(X; 0, \lambda, \theta), Y(X; 1, \lambda, \theta) , \ldots).
\end{equation}
 We have that
 $$\|Y(X;  \lambda, \theta)\|_w \leq \|\phi(X;  \lambda, \theta)\|_w +|\psi_0(X; \lambda, \theta)| \, \|\mu_0(\cdot; \lambda, \theta)\|_w+
 \sum_{j=1}^p|\psi_j(X; \lambda, \theta)| \, \|\mu_j(\cdot; \lambda, \theta)\|_w.$$
As a consequence of Lemma \ref{lemma1}, {\eqref{caca0}, \eqref{cacai}}, Lemma \ref{lemma2} (ii), Assumption \ref{estim}, and the above inequality, we have that  $E_{\lambda, \theta} \left\{\|Y(X;  \lambda, \theta)\|_w^2 \right \}<\infty$. Now, by the central limit theorem in Hilbert spaces (see Theorem 2.7 in Bosq \cite{Bosq}),
\[
  \frac{1}{\sqrt{n}}\sum_{i=1}^{n}Y(X_i;\lambda, \theta)\stackrel{\mathcal{L}}{\longrightarrow}G(\lambda, \theta).
  \]
 The result follows from the previous relations and the continuous mapping theorem.  
\end{proof}

\begin{remark} \label{cova}
The covariance kernel $C$ of $G(\lambda,\theta)$ has the following expression,
\begin{eqnarray*}
C(k,s) & = & E_{\lambda,\theta}\{\phi(X; k, \lambda, \theta) \phi(X; s, \lambda, \theta)\} + \sum_{j=0}^p\mu_j(k;  \lambda, \theta)
 E_{\lambda,\theta}\{\phi(X; s, \lambda, \theta)  \psi_j(X; \lambda, \theta)\}\\
  & & +\sum_{j=0}^p\mu_j(s;  \lambda, \theta)
 E_{\lambda,\theta}\{\phi(X; k, \lambda, \theta)  \psi_j(X; \lambda, \theta)\}\\
  & & + \sum_{j,v=0}^p\mu_j(k;  \lambda, \theta) \mu_v(s;  \lambda, \theta)
 E_{\lambda,\theta}\{ \psi_j(X; \lambda, \theta) \psi_v(X; \lambda, \theta)\},
\end{eqnarray*}
with
\begin{eqnarray*}
E_{\lambda,\theta}\{\phi(X; k, \lambda, \theta) \phi(X; s, \lambda, \theta)\}  & = & \left\{
\begin{array}{ll}
\displaystyle (k+1)^2p_{k+1}+\lambda^2\sum_{u=0}^kq_{k-u} (\theta)^2p_u & \mbox{if } k=s,\\
\displaystyle  -\lambda (k+1)q_{s-k-1}(\theta)p_{k+1}+\lambda^2\sum_{u=0}^kq_{k-u} (\theta) q_{s-u} (\theta)p_u & \mbox{if } k<s,
\end{array}
\right. \\
E_{\lambda,\theta}\{\phi(X; k, \lambda, \theta)  \psi_j(X; \lambda, \theta)\} & = & (k+1)p_{k+1}\psi_j(k+1; \lambda, \theta)-\lambda
\sum_{u=0}^kq_{k-u} (\theta) p_u \psi_j(u; \lambda, \theta),
\end{eqnarray*}
$k,s \in \mathbb{N}_0$, $0\leq j \leq p$.
\end{remark}

\begin{remark} \label{probremark}
Assume that $H_0$ is true, that is,  $X\sim {GPD}(\lambda,\theta,1,G_{21}(\cdot; \theta))$, for some $\lambda>0$ and some  $\theta \in \Theta$.
If Assumption \ref{Q}(i) holds, from \eqref{Q.21}, by dominated convergence theorem, we can write
$$ \int {G}_{21}(t; \theta)dt =\sum_{k\geq 0}q_k(\theta)\frac{t^{k+1}}{k+1}+\mbox{constant}, \quad t\in[0,1].$$
Thus, from Remark \ref{charact}, we \red{have} that the pgf of $X$, $g(t)$, has the following expression:
$$g(t)=\exp\Big(\lambda \sum_{k\geq 0}q_k(\theta)\frac{t^{k+1}-1}{k+1}\Big).$$
Hence,
$$p_0=g(0)=\exp\Big(-\lambda \sum_{k\geq 0}q_k(\theta)\frac{1}{k+1}\Big).$$
Notice that if  Assumption \ref{Q}(ii) and (iii) hold, then it can be easily  checked that $p_0$ is  a continuous function of $\lambda$ and $\theta$. From \eqref{relation}, this fact implies that  $p_k$ is  a continuous function of $\lambda$ and $\theta$, $\forall k \geq 0$.
\end{remark}

\red{Before proving Theorem \ref{boot}, we state a preliminary result.}

\red{\begin{lemma} \label{CC}  Let $\{\lambda_n\}_{n \geq 1}$  and $\{\theta_n\}_{n \geq 1}$ be two sequences.
Assume that the weights  $w_0,w_1,\ldots$ satisfy \eqref{w}, that $E_{\lambda,\theta}(X^2)<\infty$, that Assumptions \ref{Q} and \ref{estim2} hold. Let $C(k,k)$ denote the covariance kernel in Remark  \ref{cova} with $k=s$ and let $C_n(k,k)$  denote  $C(k,k)$ with $\lambda$ and $\theta$  replaced with $\lambda_n$ and $\theta_n$, respectively. If $\lambda_n \to \lambda>0$  and $\theta_n \to \theta \in \Theta$, then  $\sum_{k \geq 0} w_k \left\{ C_n(k,k)-C(k,k) \right\} \to 0$.
\end{lemma}}

\begin{proof} \red{From Remark \ref{cova} we can write
$\sum_{k \geq 0} w_k \left\{ C_n(k,k)-C(k,k) \right\}=S_1+S_2+2S_{30}+\ldots +2S_{3p}{-2S_{40}-\ldots -2S_{4p}}+S_{500}+\ldots+S_{5pp}$, the expression of each term on the right-hand side of the previous equality will be given a bit later. To prove the result, it will be seen that each of those terms go to 0. Let $p_{k}(\lambda,\theta)=P_{\lambda,\theta}(X=k)$.}

\red{
\underline{Term $S_1$}: {this term has} the following expression
$$S_1=\sum_{k \geq 0} w_k(k+1)^2 \left\{ p_{k+1}(\lambda_n,\theta_n) -p_{k+1}(\lambda,\theta)\right\}.$$
As $E_{\lambda,\theta}(X^2)<\infty$, for any $\varepsilon>0$, there exits $k_0=k_0(\varepsilon, \lambda, \theta) \in \mathbb{N}$ such that $\sum_{k > k_0} k^2p_{k}(\lambda,\theta)<\varepsilon$.
 Because $E_{\lambda_n,\theta_n}(X^2)-E_{\lambda,\theta}(X^2)\to 0$ (see the comments after Assumption \ref{Q}), we have that for large enough $n$, $|E_{\lambda_n,\theta_n}(X^2)-E_{\lambda,\theta}(X^2)|<\varepsilon$. Hence, for large enough $n$,
\begin{eqnarray*}
  \sum_{k > k_0} k^2p_{k}(\lambda_n,\theta_n)  & = & \sum_{k > k_0} k^2p_{k}(\lambda,\theta)+
  \sum_{k > k_0} k^2\left\{ p_{k}(\lambda_n,\theta_n)-p_{k}(\lambda,\theta)\right\}\\
   & = & \sum_{k > k_0} k^2p_{k}(\lambda,\theta)+E_{\lambda_n,\theta_n}(X^2)-E_{\lambda,\theta}(X^2)+\sum_{k=1}^{k_0}k^2\left\{p_{k}(\lambda,\theta)- p_{k}(\lambda_n,\theta_n)\right\}\\
   & \leq & \varepsilon+\varepsilon+k_0^2 \max_{1\leq k \leq k_0} \left | p_{k}(\lambda,\theta)- p_{k}(\lambda_n,\theta_n) \right |<3\varepsilon,
\end{eqnarray*}
since from Remark \ref{probremark},  $\displaystyle k_0^2 \max_{1\leq k \leq k_0} \left | p_{k}(\lambda,\theta)- p_{k}(\lambda_n,\theta_n) \right |  \leq \varepsilon$, for large enough $n$. 
Finally,  for large enough $n$,
\begin{eqnarray*}
\left | S_1 \right |  & = & \left| \sum_{k \leq  k_0} w_k(k+1)^2 \left\{ p_{k+1}(\lambda_n,\theta_n) -p_{k+1}(\lambda,\theta)\right\}+  \sum_{k > k_0} w_k(k+1)^2 \left\{ p_{k+1}(\lambda_n,\theta_n) -p_{k+1}(\lambda,\theta)\right\} \right| \\
& \leq  & M k_0^2 \max_{1\leq k \leq k_0} \left | p_{k}(\lambda,\theta)- p_{k}(\lambda_n,\theta_n) \right |+ M  \sum_{k > k_0} k^2p_{k}(\lambda_n,\theta_n)+ M  \sum_{k > k_0} k^2p_{k}(\lambda,\theta)\\
 & \leq & M(\varepsilon+\varepsilon+3\varepsilon)=5M\varepsilon,
\end{eqnarray*}
which proves that $S_1 \to 0$.}

\red{
\underline{Term $S_2$}: {this term has} the following expression
\begin{eqnarray*}
S_2 & = & \sum_{k \geq 0} w_k \left\{
\lambda_n^2 \sum_{u= 0}^k q_u^2(\theta_n)p_{k-u}(\lambda_n,\theta_n) 
-\lambda^2 \sum_{u= 0}^k q_u^2(\theta)p_{k-u}(\lambda,\theta) \right\}\\
& = & (\lambda_n^2-\lambda^2)  \sum_{k \geq 0} w_k \sum_{u= 0}^k q_u^2(\theta)p_{k-u}(\lambda,\theta) +\lambda_n^2 \sum_{k \geq 0} w_k \left\{
 \sum_{u= 0}^k q_u^2(\theta_n)p_{k-u}(\lambda_n,\theta_n) -
 \sum_{u= 0}^k q_u^2(\theta)p_{k-u}(\lambda,\theta) \right\}\\
 &:= & S_{21}+S_{22}.
\end{eqnarray*}
For  $S_{21}$ we have that
$$|S_{21}| \leq  M \left | \lambda_n^2-\lambda^2 \right | \sum_{k \geq 0}  \sum_{u= 0}^k q_u^2(\theta)p_{k-u}(\lambda,\theta)=  M \left | \lambda_n^2-\lambda^2 \right | \sum_{k \geq 0} q_u^2(\theta).
$$
{By Assumption \ref{Q}} $\sum_{k \geq 0}\left| q_k(\theta) \right |<\infty$, and therefore we also have that 
\begin{equation} \label{auxq1}
\sum_{k \geq 0} q_k^2(\theta)<\infty.
\end{equation}
Therefore,  $S_{21} \to 0$.}

\red{
From \eqref{auxq1}, we have that for any fixed $\varepsilon>0$ there exists 
 $k_0=k_0(\varepsilon,  \theta) \in \mathbb{N}$ such that $\sum_{k > k_0}  q_u^2(\theta)<\varepsilon$. On the other hand,
\begin{eqnarray*}
  \left | \sum_{k \geq 0} \left\{ q_k^2(\theta) - q_k^2(\theta_n) \right\}\right|  &  \leq &  
   \sum_{k \geq 0}  \left |q_k(\theta) + q_k(\theta_n)  \right|  \left |q_k(\theta) -q_k(\theta_n)  \right| \\
   & \leq & \left\{  \sum_{k \geq 0} q_k(\theta) +  \sum_{k \geq 0} q_k(\theta_n)\right\}  \sum_{k \geq 0}  \left |q_k(\theta) -q_k(\theta_n)  \right| \\
   & \leq & \left\{  2\sum_{k \geq 0} q_k(\theta) + \sum_{k \geq 0}  \left |q_k(\theta) -q_k(\theta_n)  \right| \right\}  \sum_{k \geq 0}  \left |q_k(\theta) -q_k(\theta_n)  \right|,
\end{eqnarray*}
and hence, by Assumption \ref{Q} (i) and (iv), it follows that
\begin{equation} \label{auxq2}
\sum_{k \geq 0} \left\{ q_k^2(\theta) - q_k^2(\theta_n) \right\}\to 0.
\end{equation}
Now, proceeding as in the proof for $S_1$, one gets that $S_{22} \to 0$.}

\red{
\underline{Term $S_{3j}$, $0\leq j \leq p$}: {this term has} the following expression
\begin{eqnarray*}
S_{3j} & = & \sum_{k \geq 0} w_k (k+1)\left \{ \mu_j(k;  \lambda_n, \theta_n) p_{k+1}(\lambda_n, \theta_n)  \psi_j(k+1; \lambda_n, \theta_n)- \mu_j(k;  \lambda, \theta) p_{k+1}(\lambda, \theta)  \psi_j(k+1; \lambda, \theta) \right\}\\
 & = &   \sum_{k \geq 0} w_k \left\{ \mu_j(k;  \lambda_n, \theta_n)-\mu_j(k;  \lambda, \theta) \right\}(k+1)p_{k+1}(\lambda_n, \theta_n)  \psi_j(k+1; \lambda_n, \theta_n)\\
  & &
+\sum_{k \geq 0} w_k \mu_j(k;  \lambda, \theta)  (k+1) \left\{p_{k+1}(\lambda_n, \theta_n)  \psi_j(k+1; \lambda_n, \theta_n)-p_{k+1}(\lambda, \theta)  \psi_j(k+1; \lambda, \theta)\right\}\\
 &:= & S_{31j}+S_{32j}.
\end{eqnarray*}
For  $S_{31j}$ we have that
\begin{eqnarray*}
|S_{31j}| & \leq &  M \sup_{k \geq 0} \left |  \mu_j(k;  \lambda_n, \theta_n)-\mu_j(k;  \lambda, \theta) \right | \sum_{k \geq 0}    (k+1) p_{k+1}(\lambda_n, \theta_n) \left |  \psi_j(k+1; \lambda_n, \theta_n) \right | \\
 & \leq &  M \left[ \sum_{k \geq 0} \left \{  \mu_j(k;  \lambda_n, \theta_n)-\mu_j(k;  \lambda, \theta) \right \}^2 \right]^{1/2} E^{1/2}_{\lambda_n,\theta_n}(X^2) E^{1/2}_{\lambda_n,\theta_n}\left\{\psi_j(X; \lambda_n, \theta_n)^2 \right\}.
\end{eqnarray*}
As seen before, $E_{\lambda_n,\theta_n}(X^2) \to E_{\lambda,\theta}(X^2)<\infty$. From Assumption \ref{estim2} (ii), $E^{1/2}_{\lambda_n,\theta_n}\left\{\psi_j(X; \lambda_n, \theta_n)^2 \right\}
\to E^{1/2}_{\lambda,\theta}\left\{\psi_j(X; \lambda, \theta)^2 \right\}<\infty$. For $j=0$,
\begin{eqnarray*}
 0.5\sum_{k \geq 0} \left \{  \mu_0(k;  \lambda_n, \theta_n)-\mu_0(k;  \lambda, \theta) \right \}^2 & \leq &   \sum_{k \geq 0} \left [ \sum_{u=0}^k q_{u}(\theta) \left\{p_{k-u}(\lambda_n, \theta_n)-p_{k-u}(\lambda, \theta) \right\} \right]^2\\ &  & +
  \sum_{k \geq 0} \left [ \sum_{u=0}^k p_{u}(\lambda_n, \theta_n) \left\{q_{k-u}(\theta_n)-q_{k-u}(\theta) \right\} \right ]^2\\
  & := &  D_1+D_2.
\end{eqnarray*}  
We have that
$$ 
D_1  =  \sum_{u,v \geq 0} q_u(\theta)q_v(\theta) \sum_{k \geq 0} \left\{p_{k}(\lambda_n, \theta_n)-p_{k}(\lambda, \theta) \right\} \left\{p_{k+|u-v|}(\lambda_n, \theta_n)-p_{k+|u-v|}(\lambda, \theta) \right\}.
$$
For each fixed $u,v \geq 0$,
\[
\sum_{k \geq 0} \left\{p_{k}(\lambda_n, \theta_n)-p_{k}(\lambda, \theta) \right\} \left\{p_{k+|u-v|}(\lambda_n, \theta_n)-p_{k+|u-v|}(\lambda, \theta) \right\} \leq 
\]
\[
\left [\sum_{k \geq 0} \left\{p_{k}(\lambda_n, \theta_n)-p_{k}(\lambda, \theta) \right\}^2
\sum_{k \geq 0}  \left\{p_{k+|u-v|}(\lambda_n, \theta_n)-p_{k+|u-v|}(\lambda, \theta) \right\}^2 \right ]^{1/2} \leq  
\]
\[
\sum_{k \geq 0} \left\{p_{k}(\lambda_n, \theta_n)-p_{k}(\lambda, \theta) \right\} ^2,
\]
and thus
$$D_1 \leq \left \{ \sum_{k \geq 0} \left| q_k(\theta) \right |  \right\}^2 \sum_{k \geq 0} \left\{p_{k}(\lambda_n, \theta_n)-p_{k}(\lambda, \theta) \right\} ^2.
$$ 
 Now proceeding as in the proof for $S_1$ it can be seen that  $\sum_{k \geq 0} \left\{p_{k}(\lambda_n, \theta_n)-p_{k}(\lambda, \theta) \right\} ^2 \to 0$, which proves that $D_1\to 0$, because 
 $\sum_{k \geq 0} \left| q_k(\theta) \right | <\infty$ by Assumption \ref{Q}(i).}

\red{
As for $D_2$, taking into account  \eqref{tia}, we have that
\begin{eqnarray*}
D_2 & = & \sum_{u,v \geq 0} \left \{q_{u}(\theta_n)-q_{u}(\theta) \right\} \left \{
q_{v}(\theta_n)-q_{v}(\theta) \right\}
 \sum_{k \geq 0} p_k(\lambda_n, \theta_n) p_{k+|u-v|}(\lambda_n, \theta_n)\\
& \leq & \left\{ \sum_{k \geq 0}
\left | q_{k}(\theta_n)-q_{k}(\theta) \right |  \right\}^2,
\end{eqnarray*}
which, by Assumption \ref{Q} (iv), implies that  $D_2\to 0$. So we have proven that $S_{31j} \to 0$ for $j=0$. For $1 \leq j \leq p$, the proof is  similar, and thus omitted.}

\red{Now we deal with $S_{32j}$,
\begin{eqnarray*}
|S_{32j}| & \leq & \sqrt{M} \sup_{k \geq 0} \sqrt{w_k} \left|  \mu_j(k;  \lambda, \theta)\right |  \sum_{k \geq 0} (k+1) \left | p_{k+1}(\lambda_n, \theta_n)  \psi_j(k+1; \lambda_n, \theta_n)-p_{k+1}(\lambda, \theta)  \psi_j(k+1; \lambda, \theta) \right | \\
          & \leq & \sqrt{M} \left \|  \mu_j( \cdot;  \lambda, \theta) \right \|_w  \sum_{k \geq 0} (k+1) \left | p_{k+1}(\lambda_n, \theta_n)  \psi_j(k+1; \lambda_n, \theta_n)-p_{k+1}(\lambda, \theta)  \psi_j(k+1; \lambda, \theta) \right |.
\end{eqnarray*}
Notice  that  $\left \|  \mu_0( \cdot;  \lambda, \theta) \right \|_w = \left \|E\left\{\frac{\partial}{\partial \lambda}\widehat{d}(\cdot; \lambda,\theta)\right\} \right \|_w^2<\infty$, and   $\left \|  \mu_j( \cdot;  \lambda, \theta) \right \|_w =\left \|E\left\{\frac{\partial}{\partial \theta_j}\widehat{d}(\cdot;\lambda,\theta)\right\}\right\|_w^2<\infty$, $1 \leq j \leq p$. 
The finiteness of those quantities comes from Lemma \ref{lemma2} (ii). {Taking into account that $E_{\lambda, \theta}\{X\psi_j(X;\lambda,\theta)\} \leq E^{1/2}_{\lambda, \theta}(X^2) E^{1/2}_{\lambda, \theta}\{\psi_j(X;\lambda,\theta)^2\}<\infty$  and Assumption \ref{estim2} (iv),}
proceeding as in the proof for $S_1$ it can be seen that 
 $\sum_{k \geq 0} (k+1) \left | p_{k+1}(\lambda_n, \theta_n)  \psi_j(k+1; \lambda_n, \theta_n)-p_{k+1}(\lambda, \theta)  \psi_j(k+1; \lambda, \theta) \right | \to 0$, which proves that $S_{32j}\to 0$, $0\leq j \leq p$.
}

{
\underline{Term $S_{4j}$, $0\leq j \leq p$}: {this term has} the following expression
\begin{eqnarray*}
S_{4j} & = & \sum_{k \geq 0} w_k \left \{ \mu_j(k;  \lambda_n, \theta_n) \lambda_n \sum_{u=0}^kq_{k-u}(\theta_n) p_u(\lambda_n, \theta_n)  \psi_j(u; \lambda_n, \theta_n) \right. \\
 & & - \left .
 \mu_j(k;  \lambda, \theta) \lambda \sum_{u=0}^kq_{k-u}(\theta) p_u(\lambda, \theta)  \psi_j(u; \lambda, \theta) \right\}\\
 & = &   (\lambda_n-\lambda)  \sum_{k \geq 0} w_k \mu_j(k;  \lambda_n, \theta_n)  \sum_{u=0}^kq_{k-u}(\theta_n) p_u(\lambda_n, \theta_n)  \psi_j(u; \lambda_n, \theta_n)\\
  & & + \lambda  \sum_{k \geq 0} w_k \{ \mu_j(k;  \lambda_n, \theta_n) -\mu_j(k;  \lambda, \theta)  \} \sum_{u=0}^k q_{k-u}(\theta_n) p_u(\lambda_n, \theta_n)  \psi_j(u; \lambda_n, \theta_n)
 \\
 & & +\lambda  \sum_{k \geq 0} w_k  \mu_j(k;  \lambda, \theta)   \sum_{u=0}^k \{q_{k-u}(\theta_n) -q_{k-u}(\theta)\}p_u(\lambda_n, \theta_n)  \psi_j(u; \lambda_n, \theta_n) \\
  & & +\lambda  \sum_{k \geq 0} w_k  \mu_j(k;  \lambda, \theta)   \sum_{u=0}^k q_{k-u}(\theta) \{ p_u(\lambda_n, \theta_n)  \psi_j(u; \lambda_n, \theta_n) - p_u(\lambda, \theta)  \psi_j(u; \lambda, \theta)\}\\
 &:= & S_{41j}+S_{42j}+S_{43j}+S_{44j}.
\end{eqnarray*}
We have the following facts:
\begin{itemize} \itemsep=0pt
\item[F1] $\displaystyle
\left|   \sum_{u=0}^kq_{k-u}(\theta_n) p_u(\lambda_n, \theta_n)  \psi_j(u; \lambda_n, \theta_n)  \right| \leq 
\left(  \sum_{u=0}^k q_{k-u}^2(\theta_n) \right)^{1/2}E^{1/2}_{\lambda_n, \theta_n} \{\psi_j^2(u; \lambda_n, \theta_n)\}.$
\item[F2] From \eqref{auxq1} and  \eqref{auxq2}, it follows that, for large  enough $n$, $\sum_{u=0}^k q_{k-u}^2(\theta_n) \leq M_1$, for some finite $M_1$.
\item[F3] By Assumption \ref{estim2} (ii), 
$E_{\lambda_n,\theta_n}\{
\psi_j^2(X;\lambda_n,\theta_n)\} \to E_{\lambda,\theta}\{\psi_j^2(X; \lambda,\theta) \}<\infty$, $0\leq j \leq p$.
\item[F4] In the proof for $S_{31j}$ we saw that $\| \mu_j(k;  \lambda_n, \theta_n)- \mu_j(k;  \lambda, \theta)\|_w\to 0$, and in the proof for $S_{32j}$ we saw that $\|\mu_j(k;  \lambda, \theta)\|_w<\infty$, $0\leq j \leq p$.
\end{itemize}
All these facts imply that $S_{41j} \to 0$ and $S_{42j} \to 0$, $0\leq j \leq p$.
}

{ 
Taking into account
\begin{eqnarray*}
\sum_{u=0}^k \{q_{k-u}(\theta_n) -q_{k-u}(\theta)\}p_u(\lambda_n, \theta_n)  \psi_j(u; \lambda_n, \theta_n) & \leq & \sup_u | q_u(\theta_n) -q_u(\theta)| \sum_{u \geq 0} p_u(\lambda_n, \theta_n) \left | \psi_j(u; \lambda_n, \theta_n) \right | \\ & \leq & \sum_{u \geq 0 } | q_u(\theta_n) -q_u(\theta)| E_{\lambda_n,\theta_n}^{1/2}\{
\psi_j^2(X;\lambda_n,\theta_n)\},
\end{eqnarray*}
F3, F4 and  Assumption \ref{Q} (iii), one gets that  $S_{43j} \to 0$, $0\leq j \leq p$.}

{By Assumption \ref{Q}(i), $\sum_{k \geq 0}|q_k(\theta)| <\infty$. Hence, $\forall \varepsilon>0$, there exists $k_0=k_0(\theta, \varepsilon)$, such that $\sum_{k > k_0 }|q_k(\theta) | <\varepsilon.$  
We also have that, by Assumption \ref{estim2} (iii) and Remark 
 \ref{probremark}, for any finite $K \subset \mathbb{N}_0$
$$\sup_{k \in K} | p_k(\lambda_n, \theta_n)  \psi_j(k; \lambda_n, \theta_n) - p_k(\lambda, \theta)  \psi_j(k; \lambda, \theta) | \to 0.$$
Thus,  for $k \leq k_0$,
$$\begin{array}{c}
   \displaystyle \sum_{u=0}^k q_{u}(\theta) \{ p_{k-u}(\lambda_n, \theta_n)  \psi_j(k-u; \lambda_n, \theta_n) - p_{k-u}(\lambda, \theta)  \psi_j(k-u; \lambda, \theta)\}  \leq  \vspace{3pt}\\
 \displaystyle \sum_{k \geq 0}|q_k(\theta)| \sup_{0 \leq u \leq k_0 } | p_u(\lambda_n, \theta_n)  \psi_j(u; \lambda_n, \theta_n) - p_u(\lambda, \theta)  \psi_j(u; \lambda, \theta) | \to 0,
\end{array} $$
and for $k>k_0$,
$$\begin{array}{c}
   \displaystyle 
 \sum_{u=0}^k q_{u}(\theta) \{ p_{k-u}(\lambda_n, \theta_n)  \psi_j(k-u; \lambda_n, \theta_n) - p_{k-u}(\lambda, \theta)  \psi_j(k-u; \lambda, \theta)\}  \leq \vspace{3pt}\\
  \displaystyle 
  \sum_{k \geq 0}|q_k(\theta)| \sup_{k-k_0 \leq u \leq k } | p_u(\lambda_n, \theta_n)  \psi_j(u; \lambda_n, \theta_n) - p_u(\lambda, \theta)  \psi_j(u; \lambda, \theta) | \vspace{3pt}\\
   \displaystyle 
 +
\sum_{k > k_0 }|q_k(\theta) | \left [  E_{\lambda_n,\theta_n}^{1/2}\{
\psi_j^2(X;\lambda_n,\theta_n)\}+  E_{\lambda,\theta}^{1/2}\{
\psi_j^2(X;\lambda,\theta)\} \right ], 
\end{array} $$
that according to the previous reasoning is a small quantity. Those facts and F4 show that  $S_{44j} \to 0$, $0\leq j \leq p$.}

\underline{Term $S_{5jv}$, $0\leq j,v \leq p$}: {this term has} the following expression
\begin{eqnarray*}
S_{5jv} & = &  \sum_{k \geq 0}w_k \left [ \mu_j(k;  \lambda_n, \theta_n)\mu_v(k;  \lambda_n, \theta_n) E_{\lambda_n,\theta_n}\{\psi_j(X;\lambda_n,\theta_n)\psi_v(X;\lambda_n,\theta_n)\} \right. \\ & & \left. -
\mu_j(k;  \lambda, \theta)\mu_v(k;  \lambda, \theta) E_{\lambda,\theta}\{\psi_j(X;\lambda,\theta)\psi_v(X;\lambda,\theta)\} 
\right ].
\end{eqnarray*}
Assumption \ref{estim2} (ii) and F4 both imply  that $S_{5jv}\to 0$, $0\leq j,v \leq p$.

\end{proof}

\begin{proof}[Proof of Theorem \ref{boot}] Let
 $Y(X; \lambda,\theta)$ be as defined in \eqref{Y} and let $\hat{d}^*(\cdot;\hat{\lambda}^*, \hat{\theta}^*)$ denote the bootstrap version of $\hat{d}(\cdot;\hat{\lambda},\hat{\theta})$. Proceeding as in the proof of Theorem \ref{asymptoticnulldistribution}, we have that
$$
\sqrt{n}\hat{d}^*(\cdot;\hat{\lambda}^*, \hat{\theta}^*)  =  \frac{1}{\sqrt{n}} \sum_{i=1}^n Y(X_i^*;\hat{\lambda}, \hat{\theta})+r^*,
$$
where
$\|r^*\|_w=o_{P*}(1)$, a.s (in probability). \red{Next we derive
 the (conditional) convergence in distribution of  $ (1/{\sqrt{n}}) \sum_{i=1}^n Y(X_i^*; \hat{\lambda}, \hat{\theta})$, and see that the limit coincides with the asymptotic distribution of $ (1/{\sqrt{n}}) \sum_{i=1}^n Y(X_i;  {\lambda}, \theta)$
 when the data come from $X\sim {GPD}(\lambda,\theta,1,G_{21}(\cdot; \theta))$. \red{Notice that in the bootstrap scheme, $X^*_1, \ldots , X^*_n$ are iid, but their common (conditional) distribution may vary when the sample size increases, because $\hat{\lambda}(X_1, \ldots,X_n), \hat{\theta}_1(X_1, \ldots,X_n), \, \ldots, \hat{\theta}_p(X_1, \ldots,X_n)$  may differ from   $\hat{\lambda}(X_1, \ldots,X_{n+1}), \hat{\theta}_1(X_1, \ldots,X_{n+1}), \, \ldots, \hat{\theta}_p(X_1, \ldots,X_{n+1})$ }. So we need to use a central limit theorem in Hilbert spaces for triangular arrays. Specifically, } we apply Theorem 1.1 in Kundu et al. \cite{Kundu2000} to derive the asymptotic distribution of $ (1/{\sqrt{n}}) \sum_{i=1}^n Y(X_i^*; \hat{\lambda}, \hat{\theta})$.
 With this aim, first  we  must choose an orthonormal basis of $l^2_w$.
Let $e_k=(e_k(0), e_k(1), \ldots)$, with $e_k(j)=(1/\sqrt{w_k})I(k=j)$, $k,j\geq0$. $\{e_k\}_{k\geq 0}$  is an orthonormal basis of $l^2_w$.
 \red{Next we check that all conditions in Theorem 1.1 in \cite{Kundu2000} hold true for the conditional distribution of $ Y(X_1^*; \hat{\lambda}, \hat{\theta}) \ldots,  Y(X_i^*; \hat{\lambda}, \hat{\theta})$ given the data, a.s. or in probability according to the law generating the data $X_1, \ldots, X_n$. This will be denoted as a.s. ($X_1, \ldots, X_n$) or in probability ($X_1, \ldots, X_n$)}. Bear in mind that what that paper denotes {as} $X_{ni}$ corresponds to $ Y(X^*_i;\hat{\lambda},\hat{\theta})/ \sqrt{n}$ in our setting.
\begin{itemize}
\item $E_*\{\langle Y(X^*;\hat{\lambda},\hat{\theta}),e_k \rangle_w\}=
\sqrt{w_k}E_*\{Y(X^*;k, \hat{\lambda}, \hat{\theta})\}=0$, $\forall k \geq 0$. \red{This fact happens with probability 1 ($X_1, \ldots, X_n$). This proves that the first assumption in the statement of Theorem  1.1 in \cite{Kundu2000} holds with probability 1 ($X_1, \ldots, X_n$).
 }
 
\item Proceeding as in the proof of Theorem \ref{asymptoticnulldistribution} to show that $E_{\lambda, \theta} \{ \|Y(X;\lambda, {\theta}) \|^2_w\}<\infty$, it can be also checked that $ E_*\{ \|Y(X_i^*; \hat{\lambda},\hat{\theta}) \|^2_w\} <\infty$ \red{ a.s. ($X_1, \ldots, X_n$) (in probability ($X_1, \ldots, X_n$)).}  This proves that the other assumption in the statement of Theorem  1.1  of \cite{Kundu2000} \red{holds a.s. ($X_1, \ldots, X_n$) (in probability ($X_1, \ldots, X_n$)).}

\item Let $C_n$ denote the covariance kernel, in the conditional distribution given the data,  of $Y(X^*;\hat{\lambda},\hat{\theta})$, that is, $C_n(k,s)=Cov_*\{Y(X^*;k, \hat{\lambda},\hat{\theta}), Y(X^*;s, \hat{\lambda},\hat{\theta})\}$, and let $\mathcal{C}_n$ the associated covariance ope\-ra\-tor, defined from $l_w^2$ to $l_w^2$ by $\mathcal{C}_n\ell=E_*\{ \langle Y(X^*; \hat{\lambda},\hat{\theta}), \ell \rangle_wY(X^*; \hat{\lambda},\hat{\theta}) \}$, or equivalently, $\mathcal{C}_n\ell(s)=\sum_{k \geq 0}w_k C_n(k,s) \ell(k)$. Then,
\[
 \langle \mathcal{C}_n e_k,e_r\rangle_w  =
 \sqrt{w_k} \sqrt{w_r} C_n(k,s).
 \]
 $C_n(k,s)$ has the same expression as  $C(k,s)$, the covariance kernel appearing in Theorem \ref{asymptoticnulldistribution},  with $\lambda$ and $\theta$ replaced with $\hat{\lambda}$ and $\hat{\theta}$. The expression of
 $C(k,s)$ has been given in Remark \ref{cova}. With the assumptions made, taking into account Remark \ref{probremark}, all quantities in the expression of $C(k,s)$ are continuous functions of $\lambda$ and $\theta$. Therefore,
 \[
 \langle \mathcal{C}_n e_k,e_r\rangle_w \stackrel{ }{\longrightarrow} \sqrt{w_k} \sqrt{w_r} C(k,s)=\langle \mathcal{C} e_k,e_r\rangle_w, \quad \forall k, r \geq 0,
 \]
 \red{ a.s. ($X_1, \ldots, X_n$) (in probability ($X_1, \ldots, X_n$)),}
 where $\mathcal{C}$ denotes the covariance operator associated to the covariance kernel $C$. 
 This proves that condition (i) in Theorem  1.1 of \cite{Kundu2000} \red{ holds a.s. ($X_1, \ldots, X_n$) (in probability ($X_1, \ldots, X_n$)), with $a_{kr}= \langle \mathcal{C} e_k,e_r\rangle_w$.}

 \item \red{By Lemma \ref{CC}, $\sum_{k \geq 0}w_k  C(k,k)$ is a continuous function of the parameters $(\lambda, \theta)$, therefore} %
 $$\sum_{k \geq 0}  \langle \mathcal{C}_n e_k,e_k\rangle_w=\sum_{k \geq 0}w_k  C_n(k,k) \stackrel{ }{\longrightarrow} \sum_{k \geq 0}w_k  C(k,k)= \sum_{k \geq 0} a_{kk}=
 E_{\lambda,\theta} \{ \|Y(X;{\lambda}, {\theta}) \|^2_w\}<\infty,$$
 \red{where the convergence is  a.s. ($X_1, \ldots, X_n$) (in probability ($X_1, \ldots, X_n$)), and the last equality comes from Parseval's identity.}
This proves that condition (ii) in Theorem  1.1 of \cite{Kundu2000} \red{ holds a.s. ($X_1, \ldots, X_n$) (in probability ($X_1, \ldots, X_n$)).}

 \item Finally, we must show that $L(\varepsilon,e_k)\to 0$, $\forall \varepsilon>0$, $\forall k \geq 0$, where
\begin{eqnarray*}
L(\varepsilon,e_k) & = & \frac{1}{n}\sum_{i=1}^n E_* \left\{ \langle Y(X^*_i;\hat{\lambda},\hat{\theta}),e_k\rangle_w^2 I\left[ \langle Y(X^*_i;\hat{\lambda},\hat{\theta}),e_k\rangle_w^2>\varepsilon n \right ] \right\}\\
 & = & \sqrt{w_k}  E_* \left\{ Y(X^*;k,\hat{\lambda},\hat{\theta})^2 I\left[ Y(X^*;k,\hat{\lambda},\hat{\theta})^2>\varepsilon n /\sqrt{w_k} \right ] \right\}.
\end{eqnarray*}
Let us fix $k \in \mathbb{N}_0$ and recall  \eqref{Y.expression}. We have that
$$\phi(X^*; k,\hat{\lambda},\hat{\theta})=\left\{ \begin{array}{ll} -\hat{\lambda} q_{k-u}(\hat{\theta}) & \mbox{ if $X^*=u$, $0\leq u \leq k$,}\\ k+1 & \mbox{ if $X^*=k+1$,}\\ 0 & \mbox{ if $X^*>k+1$,}\\
\end{array} \right. $$
\red{and thus,
$  \left | \phi(X^*; k,\hat{\lambda},\hat{\theta}) \right | \leq 
\max\{k+1, \hat{\lambda} q_{k-u}(\hat{\theta}) \}$. As we are assuming that  $\hat{\lambda} \stackrel{a.s. (P)}{\longrightarrow} \lambda$,
$\hat{\theta} \stackrel{a.s.(P)}{\longrightarrow} \theta$, and that Assumption \ref{Q} (i) and (iv) hold true, it follows that  $ \left | \phi(X^*; k,\hat{\lambda},\hat{\theta}) \right | $  is bounded   a.s. ($X_1, \ldots, X_n$) (in probability ($X_1, \ldots, X_n$)). On the other hand, from the proof of Lemma \ref{CC}, it follows that $\mu_j(k, \hat{\lambda},\hat{\theta})  \stackrel{ }{\longrightarrow} \mu_j(k, {\lambda},{\theta})$  a.s. ($X_1, \ldots, X_n$) (in probability ($X_1, \ldots, X_n$)), and from the proof of Lemma \ref{lemma2} (ii) we have that $\left|\mu_j(k, {\lambda},{\theta})\right|<\infty$, $0 \leq j \leq p$.
Summarizing,}
$$Y(X^*;k,\hat{\lambda},\hat{\theta})\leq M \Bigg(1+\sum_{j=0}^p |\psi_j(X^*; \hat{\lambda},\hat{\theta})| \Bigg),$$
\red{a.s. ($X_1, \ldots, X_n$) (in probability ($X_1, \ldots, X_n$)).} Now, from Assumption \ref{estim2} (v), it follows that $L(\varepsilon,e_k)\stackrel{ 
 }{\longrightarrow} 0$, $\forall \varepsilon>0$,  a.s. ($X_1, \ldots, X_n$) (in probability ($X_1, \ldots, X_n$)).
This proves that condition (iii) in Theorem  1.1 of \cite{Kundu2000} \red{ holds a.s. ($X_1, \ldots, X_n$) (in probability ($X_1, \ldots, X_n$)).}
\end{itemize}
Thus, $(1/{\sqrt{n}}) \sum_{i=1}^n Y(X_i^*;\hat{\lambda}, \hat{\theta}) \stackrel{\mathcal{L}}{\longrightarrow}G(\lambda, \theta)$ \red{  a.s. ($X_1, \ldots, X_n$) (in probability ($X_1, \ldots, X_n$)).} Now, reasoning as in the proof of Theorem \ref{asymptoticnulldistribution}, the result follows.
 \end{proof}
 

\end{document}